\documentclass[11pt]{article}
\usepackage{amsmath,amssymb,amsfonts} 
\usepackage{epsfig} \usepackage{latexsym,nicefrac,bbm}
\usepackage{xspace}
\usepackage{color,fancybox,graphicx,subfigure,fullpage}
\usepackage[top=1in, bottom=1in, left=1in, right=1in]{geometry}
\usepackage{tabularx}
\usepackage{hyperref} 
\usepackage{pdfsync}
\usepackage[boxruled]{algorithm2e}
\usepackage{multicol}
\usepackage{enumitem}
\usepackage{multirow}
\usepackage{url}
\usepackage{bm}

\renewcommand{\epsilon}{\varepsilon}

\newcommand\N{\mathbb N}
\newcommand\R{\mathbb R}

\newtheorem{theorem}{Theorem}[section]

\newtheorem{definition}{Definition}[section]
\newtheorem{lemma}[theorem]{Lemma}
\newtheorem{remark}[theorem]{Remark}
\newtheorem{proposition}[theorem]{Proposition}
\newtheorem{corollary}[theorem]{Corollary}

\newtheorem{example}[theorem]{Example}

\newenvironment{proof}{\begin{trivlist} \item {\bf Proof:~~}}
  {\qed\end{trivlist}}

\def\FullBox{\hbox{\vrule width 6pt height 6pt depth 0pt}}

\def\qed{\ifmmode\qquad\FullBox\else{\unskip\nobreak\hfil
\penalty50\hskip1em\null\nobreak\hfil\FullBox
\parfillskip=0pt\finalhyphendemerits=0\endgraf}\fi}

\newcommand{\C}{\mathbb{C}}
\renewcommand{\R}{\mathbb{R}}
\renewcommand{\N}{\mathbb{N}}

\DeclareMathOperator{\Tr}{Tr}

\DeclareMathOperator{\diag}{diag}

\DeclareMathOperator{\Stab}{Stab}
\DeclareMathOperator{\Ad}{Ad}
\DeclareMathOperator{\ad}{ad}

\DeclareMathOperator{\hull}{hull}

\DeclareMathOperator{\eig}{eig}

\def\SU{{\rm SU}}\def\U{{\rm U}}\def\SO{{\rm SO}}\def\USp{{\rm USp}}
\def\GL{{\rm GL}}\def\SL{{\rm SL}}\def\Sp{{\rm Sp}}

\usepackage[T1]{fontenc}
\usepackage{lmodern}

\title{Optimization and Sampling Under Continuous Symmetry: \\ Examples and Lie Theory  }
\author{Jonathan Leake and Nisheeth K. Vishnoi }

\begin{document}

\maketitle

\begin{abstract}
In the last few years, the notion of symmetry has provided a powerful and essential lens to view several optimization or sampling problems that arise in areas such as theoretical computer science, statistics, machine learning, quantum inference, and privacy.
Here, we 
present two examples of nonconvex problems in  optimization and sampling where continuous symmetries play -- implicitly or explicitly -- a key role in
the development of
efficient algorithms. 
These examples rely on deep and hidden connections between nonconvex symmetric manifolds and convex polytopes, and are heavily generalizable.
To formulate and understand these generalizations, we then present an introduction to Lie theory -- an indispensable mathematical toolkit for capturing and working with continuous symmetries. 
We first present the basics of Lie groups, Lie algebras, and the adjoint actions associated with them, and we also mention the classification theorem for Lie algebras.
Subsequently, we present Kostant's convexity theorem and show how it allows us to reduce linear optimization problems over orbits of Lie groups to linear optimization problems over polytopes.
Finally, we present the Harish-Chandra and the  Harish-Chandra--Itzykson--Zuber (HCIZ) formulas, which convert partition functions (integrals) over Lie groups into sums over the corresponding (discrete) Weyl groups, enabling efficient sampling algorithms.
\end{abstract}

\newpage

\tableofcontents

\newpage

\section{Introduction}

In the words of Hermann Weyl: 
``{\em A thing is symmetrical if there is something you can do to it so that after you have finished doing it, it looks the same as before.}''

\smallskip
Symmetries can be discrete or continuous.
For instance, the polynomial $f(x_1,\ldots, x_n):=x_1^2 + \cdots + x_n^2$ is symmetric under any permutation of its variables.
The associated set of symmetries is the symmetric group $S_n$.
The Boolean cube $\{0,1\}^n$ also remains invariant under the permutation of coordinates by $S_n$.
A square remains unchanged if we rotate it by multiples of $90$ degree around its center.
The integer lattice $\mathbb{Z}^n$ remains invariant under translation by any vector with integer entries.
These are examples where the set of actions that preserve the object in question is discrete -- finite or infinite.
An infinite line, e.g., $\mathbb{R}$  can be translated by any real number and it remains the same.
A sphere  can be rotated arbitrarily and it remains the same.
An $n \times n$ complex matrix $A$ can be conjugated by a unitary $U$ as $UAU^*$ and its eigenvalues remain the same.
In all of these cases, the set of actions that preserves the object is continuous and hence, infinite.

Symmetries, both discrete and continuous, play a central role in physics and mathematics.
In physics, symmetries have long served as a guiding principle to search for the laws of nature \cite{Gross14256}.
Physical scenarios are encoded using numbers, vectors, and matrices that depend on the choice of a ``reference frame,'' and we expect the laws of nature  to be invariant to changes in the reference frame. 
For instance, via the addition of a new structure -- the curvature of spacetime -- the equations of general relativity makes it possible to (approximately) describe the universe  from any reference frame or point of view using the same equations.
Interestingly, each global continuous symmetry of a physical system imposes the conservation of a quantity: symmetry under translation implies conservation of momentum, symmetry under rotation implies conservation of angular momentum, and symmetry under time implies conservation of energy. 
This is the content of Noether's theorem \cite{Noether1918}.

In mathematics, symmetries are important in almost all branches. 
In particular, Lie groups that encode continuous symmetries are central to 
analysis, topology,
algebraic geometry,
differential geometry,
number theory, and 
Riemannian geometry.
Moreover, as we explain in Sections \ref{sec:Weyl} and \ref{sec:classification}, there is deep connection between several discrete groups and Lie groups, leading to their applications in discrete mathematics and combinatorics.

Symmetries -- both discrete and continuous -- also arise in optimization and sampling problems and have been used in the design of efficient algorithms.
While there is a growing list of problems and works where the lens of symmetry has been helpful (see \cite{Barvinok05convexgeometry,sanyal2011,Waterhouse:1983:DSP,Edelman1999,Parrilo,GargGOW16,GargGOW17,SraVY18,Allen-ZhuGLOW18,BurgisserFGOWW19,LeakeV20,LeakeV20b,LMV2021} and the references therein), unlike mathematics and physics where understanding symmetries is an essential part of the basic toolkit of a student, a systematic treatment of symmetries from the point of view of optimization is lacking.   
The goal of this article is not to fill this void, but 1) to entice the reader enough to appreciate continuous symmetries and 2) to present an  introduction to Lie theory from scratch.

We start with the question of why continuous symmetries might arise in optimization and sampling problems?
The answer is often similar to the same question in physics.
Sometimes, the function that we would like to optimize  corresponds to  ``physical'' or geometric quantities such as distances, inner products, volumes, or curvature associated with points in a domain involving vectors, matrices, or tensors.
These geometric quantities often do not depend on the reference frame used to represent the underlying structures such as vectors, matrices, and tensors and, hence, continuous symmetries arise.
For instance, inner products, determinants, and trace, all have well-known symmetries.
Moreover, this set of symmetries ``compose'' with each other in a natural manner.

\newpage
Consider the function 
$ f(x_1,\ldots,x_k):= \sum_{i,j=1}^k \|x_i-x_j\|^2$, that
 is the sum of  squared Euclidean distances between $n$ points with positions $x_1,\ldots,x_k \in \mathbb{R}^n$.
This function is invariant to translating all the points by the same vector $c \in \mathbb{R}^n$: 
$$f(x_1+c,\ldots, x_n+c)=f(x_1,\ldots,x_n).$$
This set of translations corresponds to $\mathbb{R}^n$ and the translations compose as ``addition:'' Translating by $c_1$ and then by $c_2$ has the same effect as translating by $c_1+c_2$, endowing the Euclidean space with an abelian group structure.
The function $f$ is also invariant under rotations. 
Let $Q$ be an $n \times n$ orthogonal matrix, meaning that $QQ^\top = Q^\top Q = I$.
Then, 
$$f(Qx_1,\ldots,Qx_n)=f(x_1,\ldots,x_n).$$
This is because the orthogonality of $Q$ implies that
$$ \|Qx_i-Qx_j\|^2 = \langle Q(x_i-x_j), Q(x_i-x_j) \rangle = (x_i-x_j)^\top Q^\top Q (x_i-x_j)=(x_i-x_j)^\top(x_i-x_j) = \|x_i-x_j\|^2.$$
Note that the set of $n \times n$ orthogonal matrices forms a group under matrix multiplication; however, this group is nonabelian.

Moreover, sometimes, the domain itself might be symmetric.
Consider the complex unit sphere
$S_\C^n:= \left\{ (x_1,\ldots, x_n)\in \mathbb{C}^n: \sum_{i=1}^n |x_i|^2 = 1\right\}.$
For an $n \times n$ unitary matrix $U$, we have that $$U \circ S_\C^n := \{Ux: x 
\in S_\C^n\}= S_\C^n$$
because $\|Ux\|^2=\|x\|^2$.
A matrix $U \in \C^{n \times n}$ is said to be unitary if $UU^* = I$.
Hence, the action of such a unitary matrix does not change the domain.
Moreover, just like orthogonal matrices, the set of unitary matrices also forms a group under matrix multiplication.
What is special about the three groups we have seen so far? 
It is that one can do calculus on them (and hence optimize functions) and they come with natural measures associated with them and hence, we can define sampling problems on them.

In the first part of this article, we focus on two concrete examples -- one involving an optimization problem and one involving a sampling problem.
Both these results have multiple elementary proofs.
We also present elementary and self-contained proofs of these results.
The proofs of both results rely on first establishing a connection between a Lie group and a convex polytope.
And then, going from this convex polytope to a discrete group.
These proofs have been carefully chosen to set the stage for the introduction of Lie theory and, finally, presentation of deep generalizations of both these results.
\paragraph{Example 1: Minimum eigenvalue of a matrix.} The first problem we consider is the following well-known optimization characterization of the smallest eigenvalue of a Hermitian matrix.\footnote{An $n \times n$ matrix $A$ with entries in $\mathbb{C}$ is Hermitian if $A=A^*$.}
\begin{theorem}\label{thm:eigenvalue}
Let $A$ be an $n \times n$ Hermitian matrix with smallest eigenvalue $\lambda_1$, then 
\begin{equation}\label{eq:eigenvalue}
 \min _{v \in \mathbb{C}^n: \|v\|=1} v^* A v = \lambda_1.
 \end{equation}
\end{theorem}

\noindent
Neither the objective function nor the domain in Equation \eqref{eq:eigenvalue} is convex.
We start by rewriting the above problem as an optimization problem over the following (nonconvex) manifold:
\[
    \mathcal{P}_1:= \{ X: X \in \C^{n \times n}, \Tr(X)=1, X=X^*, X^2=X\}.
\]
The reason is that $\mathcal{P}_1$ consists of rank one projection matrices and the only rank one projections are of the type $vv^*$ for a vector $v$ such that $\|v\|=1$. 
Over this manifold, the objective function of Equation \eqref{eq:eigenvalue} becomes linear:
$$ \langle A,X\rangle_F.$$
Here $\langle \cdot, \cdot \rangle_F$ is the Frobenius inner product in the space of matrices.
The first observation is that while $\mathcal{P}_1$ is nonconvex, it has symmetries: 
Any $n \times n$ unitary matrix $U$  acting on $\mathcal{P}_1$ by conjugation leaves $\mathcal{P}_1$ invariant:
$$ U \mathcal{P}_1 U^* = \{ UXU^*: X \in \mathcal{P}_1 \} = \mathcal{P}_1.$$
The second observation is that the set of all $n \times n$ unitary matrices is a group under matrix multiplication: 
the {unitary group} $\U(n)$, which  is also a manifold.
Thus,  
$$\mathcal{P}_1 = \{ Ue_1e_1^*U^*: U \in \mathrm{U}(n)\}.\footnotemark{}$$
Hence, the problem in Equation \eqref{eq:eigenvalue} is not just any nonconvex problem, it can be paramterized as a linear optimization problem over a continuous group that is also a manifold; in fact, a Lie group.
This structure is implicitly or explicitly used in any proof of Theorem \ref{thm:eigenvalue}. 
In particular, this structure implies a ``convexity'' which, in turn, implies  Theorem \ref{thm:eigenvalue}.
\footnotetext{$e_1$ is the vector that is one in the first coordinate and zero elsewhere.}

In Section \ref{sec:eigenvalue}, we present a proof of this theorem from this point of view.
Where we show how this optimization problem can be reduced to a linear optimization problem over a polytope and, consequently, over its vertices. 

The reader might wonder why should we care about this point of view and such a proof.
Especially since Theorem 
\ref{thm:eigenvalue} can be proved without a reference to a group or manifold.
In short, the answer is that this viewpoint leads to a far-reaching generalization: Any linear optimization problem over such a symmetric object reduces to a linear optimization problem over an associated convex polytope; see Theorem \ref{thm:Kostant} and Corollary \ref{cor:linearopt}.
However, to state the latter result formally and understand its proof, one needs the language of Lie theory, which is discussed in Section \ref{sec:Lie}.

\paragraph{Example 2: Sampling from an exponential density on the complex unit sphere.} The second example  we consider  is the ``sampling'' version of the first example.
Given an $ n \times n$ Hermitian matrix $A$, sample from the probability distribution where a unit vector $v$ on the complex unit sphere is picked with probability proportional to $e^{-v^*Av}$.
The complex unit sphere in $n$ dimensions is the set $S_\C^n := \{(z_1,\ldots,z_n) \in \mathbb{C}^n: \sum_i |z_i|^2=1\}.$

First, to formally define this problem, we note that the complex unit sphere has a Haar measure $\mu_{S_\C^n}$ with respect to which we can define the above density.
And second, the density $e^{-v^*Av}$ favors vectors with lower values of $v^*Av$, making this problem as a ``robust'' version of the optimization problem.
Finally, note that the density can be viewed as an exponential density $e^{-\langle A,X\rangle_F}$ where  $X \in \mathcal{P}_1$.
Let $\mu_1$ be the pushforward measure of the map from the complex unit sphere to $\mathcal{P}_1$ that maps $v \mapsto vv^*$. 
$\mu_1$ is the unique unitarily invariant measure on $\mathcal{P}_1$.
Thus, we have not only a group structure, a manifold structure (as discussed in the optimization example), but also an invariant measure.

Such exponential densities arise under the name \emph{matrix Langevin} or {\em matrix Bingham} in statistics   \cite{ChikusePaper,ChikuseBook}, in quantum inference \cite{LMV2021}, and in the context of differentially private rank-$1$ approximation of a given matrix \cite{Kamalika,KTalwar}. 
Unlike the eigenvalue (optimization) problem discussed earlier, efficient sampling algorithms from such densities were only recently discovered \cite{LMV2021}, and they crucially relied on the ``symmetry viewpoint.''

In fact, a very related and almost equivalent problem is that of computing the following integral or the partition function of this exponential density:
$$ \int_{\mathcal{P}_1} e^{-\langle A,X\rangle_F}d\mu_1(X).$$
Here, we focus on this problem instead of the sampling problem.
We prove the following  formula for this integral that immediately implies an efficient algorithm to compute it.
\begin{theorem}\label{thm:partition}
    Let $A$ be an $n\times n$ Hermitian matrix with distinct eigenvalues $\lambda_1 < \cdots < \lambda_n$. 
    Then 
$$  \int_{\mathcal{P}_1} e^{-\langle A,X\rangle_F}d\mu_1(X)  = (n-1)! \sum_{i=1}^n \frac{e^{-\lambda_i }}{\prod_{j \neq i} (\lambda_j-\lambda_i)}. $$ 
\end{theorem}
It should be surprising that this integral (over the unitary group) reduces to such a formula involving a finite sum with a small number of terms.
While it is not obvious, we point out that the right hand side is in fact a determinant (of a Vandermonde-type) matrix which, in turn, is a summation over $S_n$.
Thus, an integral over $\U(n)$ reduces to a sum over the symmetric group $S_n$.
We present a self-contained proof of this theorem  in Section \ref{sec:ex2} that relies on a  connection between two highly symmetric objects -- $\mathcal{P}_1$ and (another convex polytope) the probability simplex.

Once we see a proof for the case when the eigenvalues are all distinct, it is not too difficult to come up with a similar formula in the case when eigenvalues are repeated.
Moreover, while we do not present it here, the ideas that go in the proof of this theorem can be converted into an efficient sampling algorithm; see \cite{LMV2021}.

The reader would have every reason to suspect that this result is also a tip of the iceberg: This formula is a very special case of the Harish-Chandra integral formula; see Theorems \ref{thm:HC} and \ref{thm:HCIZ}.  
While one needs the toolkit of  Lie theory to formulate and prove these results, let us conclude by mentioning that these results are remarkable: They convert certain integrals over continuous groups to summations over discrete groups -- Weyl groups.

\section{Optimization Characterization of the Minimum Eigenvalue of a Matrix}\label{sec:eigenvalue}

In this section we prove Theorem \ref{thm:eigenvalue}.
Recall that  we are given an $n \times n$ Hermitian matrix $A$ and our goal is  to prove that
$$ \min_{X \in \mathcal{P}_1} \langle A,X \rangle_F = \min_{U \in \U(n)} \langle A,Ue_1e_1^\top U^* \rangle = \lambda_1.\footnotemark{}$$
Let $A=WDW^*$ where $W$ is a unitary matrix and $D=\diag(\lambda)$ is the diagonal matrix consisting of eigenvalues of $A$.
Then we obtain\footnotetext{$e_1e_1^\top$ is a diagonal matrix with  $1$ in the $(1,1)$ entry and $0$ everywhere else.} 
 $$ \min_{U \in \U(n)} \langle WDW^*,Ue_1e_1^\top U^* \rangle_F  = \min_{U \in \U(n)} \langle U^*WDW^*U,e_1e_1^\top  \rangle_F = \min_{V \in \U(n)} \langle VDV^*,e_1e_1^\top  \rangle_F.$$
Thus, the answer to this minimization problem is the smallest entry in the diagonal of $VDV^*$ when $V$ varies over $\U(n)$.

\paragraph{From a unitary orbit to a polytope.}
We compute the diagonal entries of $A = VDV^*$ via:
    \[
        a_{ii} = \sum_{j=1}^n v_{ij} \lambda_j v_{ij}^* = \sum_{j=1}^n |v_{ij}|^2 \lambda_j = [(V \odot V^*)\lambda]_i.
    \]
  That is, the diagonal of $A$ is given by $(V \odot V^*) \lambda$, where $\odot$ denotes the entrywise product. 
   Let $Q := V \odot V^*$. 
   All entries of $Q$ are nonnegative.
   Since $V$ is a unitary matrix, we have that for all $1\leq j \leq n$,
   $$ \sum_{i=1}^n  Q_{ij} =  \sum_{i=1}^n |v_{ij}|^2 = 1.$$
   Similarly, for all $1\leq i \leq n$,
   $$ \sum_{j=1}^n  Q_{ij} =  \sum_{j=1}^n |v_{ij}|^2 = 1.$$
   Thus, $Q$ is doubly stochastic.

   The set of all doubly stochastic matrices is convex and, in fact, a polytope -- the Birkhoff polytope \cite{birk:46}.
   The well-known (and easy to prove) Birkhoff-von Neumann theorem states that the Birkhoff polytope is a convex hull of $n \times n$ permutation matrices.
   Thus, $Q$, which is doubly stochastic can be written as a convex combination of permutation matrices.
   Hence, for any unitary $V$, the diagonal of $A=VDV^*$ which is $Q\lambda$, can be written as a convex combination of permutations of $\lambda$.  

Thus, we have made a connection between the set of all Hermitian matrices with eigenvalues $\lambda=(\lambda_1,\ldots, \lambda_n)$ -- the unitary ``orbit'' of $\diag(\lambda)$ --  and the ``permutation polytope'' $P_\lambda$ of $\lambda$ defined as the convex hull of all the permutations of the vector $\lambda$:
$$ 
P_\lambda := \mathrm{hull} \{ (\lambda_{\pi(1)},\ldots,\lambda_{\pi(n)}): \pi \in S_n \}.
$$
We have argued that the diagonals of the matrices in the unitary orbit of $\diag(\lambda)$ lie in $P_\lambda$.
$$ \min_{V \in \U(n)} \langle VDV^*,e_1e_1^\top  \rangle_F \geq  \min_{v \in P_\lambda} \langle v ,e_1  \rangle= \lambda_1.$$
Moreover, plugging in the unitary $U$ which  diagonalizes $A$ and ensures that the smallest eigenvalue of $A$ appears as the $(1,1)$ entry in $VAV^*$.
Thus, we get that 
$$ \lambda_1 = \langle UAU^*, e_1e_1^\top \rangle_F \geq  \min_{V \in \U(n)} \langle UDU^*,e_1e_1^\top  \rangle_F.$$
This completes the proof of Theorem \ref{thm:eigenvalue}.

\paragraph{The Schur-Horn theorem.}
The fact that, for a given vector $\lambda$, the set of diagonal vectors of matrices in the unitary orbit of $\diag(\lambda)$ lie in $P_\lambda$ is often known as the ``Schur'' part of the Schur-Horn theorem \cite{schur1923,horn1954}.
The converse, known as the ``Horn'' part, is also true, giving us the Schur-Horn theorem; the proof below is adapted from \cite{Kadison4178}.

\begin{theorem}[Schur-Horn] \label{prop:schur_horn}
    If $D := \diag(\lambda)$ is an $n \times n$  real diagonal matrix and $U$ is an $n \times n$ unitary matrix, then the diagonal of $UDU^*$ is in  $P_\lambda$. Conversely, given any vector $v \in P_\lambda$, there exists a unitary matrix $U$ such that $UDU^*$ has diagonal vector $v$.
\end{theorem}
\begin{proof}{\bf (Sketch)}
   We have already proved the first part of this theorem.
For the other direction it suffices to show: If $w$ is the diagonal vector of some $UDU^*$ and $\tau$ is any transposition of $i,j \in [n]$, then $(1-t)w + t(\tau \cdot w)$ is the diagonal vector of $VDV^*$ for some unitary $V$ for any $t \in [0,1]$, where $\tau \cdot w$ indicates the corresponding permutation of the entries of $w$. 
    Since $w \mapsto (1-t)w + t(\tau \cdot w)$ only alters 2 entries of $w$, it suffices to demonstrate the desired claim on $2 \times 2$ matrices. To show this, we first have from the previous paragraph that $$\diag(UDU^*) = (|u_{11}|^2\lambda_1 + |u_{12}|^2\lambda_2, |u_{21}|^2\lambda_1 + |u_{22}|^2\lambda_2).$$ 
    Choosing $U := \left[\begin{smallmatrix} \sqrt{1-t} & \sqrt{t} \\ -\sqrt{t} & \sqrt{1-t} \end{smallmatrix}\right]$ for $t \in [0,1]$ gives $$\diag(UDU^*) = \left((1-t) \cdot \lambda_1 + t \cdot \lambda_2, t \cdot \lambda_1 + (1-t) \cdot \lambda_2\right).$$ 
    Letting $t$ vary over $[0,1]$ then gives all possible length-2 vectors majorized by $\lambda$.
\end{proof}

\noindent
Thus, if we define a mapping 
$$ \phi: \{U\diag(\lambda) U^*: U \in \U(n)\} \to P_\lambda$$ 
such that 
$$ \phi(U\diag(\lambda) U^*)=\diag(U\diag(\lambda) U^*),$$
then the range of this map is the convex polytope $P_\lambda$.
Note that  $P_{(1,0,\ldots,0)}$ is the probability simplex.

\paragraph{A generalization of Theorem \ref{thm:eigenvalue}: From  $ \U(n)$ to $S_n$.} We now show how to use the Schur-Horn theorem to   reduce the problem of optimizing a linear function over the infinite unitary group to the finite symmetric group.
Given two real diagonal $n \times n$ matrices $D$ and $D'$, we have that:
    \begin{equation}\label{eq:Sn}
        \min_{U \in \U(n)} \langle UDU^*, D' \rangle_F  = \min_{\sigma \in S_n} \langle \sigma D \sigma^*, D' \rangle_F.
    \end{equation}
 To see this,   let $\sigma_0$ be the permutation matrix which minimizes $\langle \sigma D \sigma^*, D' \rangle_F$ over all permutation matrices $\sigma$.\footnote{Note that all permutation matrices are unitary.}
 By the Schur-Horn theorem, for any $U$, the diagonal of $UDU^*$ can be written as a convex combination of the permutations of the diagonal of $D$. 
 By linearity of $\langle \cdot, D' \rangle_F$, the value of $\langle UDU^*, D' \rangle_F$ must then be at least the value of $\langle \sigma_0 D \sigma_0^*, D' \rangle_F$. 

\section{Sampling from an Exponential Density on the Complex Sphere}\label{sec:ex2}
In this section we prove Theorem \ref{thm:partition}.
This proof is adapted from  \cite{LeakeV20}; see also \cite{Vergne1996}.

\paragraph{From $\mathcal{P}_1$ to the simplex.}  The first step is to observe the following remarkable equality for real $A = \diag(\lambda)$, where $m$ is the Lebesgue measure on the simplex $\Delta_1:=\{p \in \mathbb{R}_+^n: \sum_{i=1}^n p_i=1\}$:
\[
    \int_{\mathcal{P}_1} e^{-\langle A, X \rangle_F} d\mu_1(X) = \int_{\Delta_1} e^{-\langle \lambda, x \rangle} dm(x).
\]
Put another way, exponential measures on $\mathcal{P}_1$, a nonconvex manifold, correspond to exponential measures on $\Delta_1$, a convex polytope.
To see this, first note the following equalities for any nonegative integers $m_1,\ldots,m_n$. The first equality is the Bombieri inner product formula (Lemma \ref{lem:bombieri}), and the second equality is a basic induction after a change of variables:
\[
    \int_{\mathcal{P}_1} X_{11}^{m_1} \cdots X_{nn}^{m_n} d\mu_1(X) = \frac{m_1! \cdots m_n! (n-1)!}{(m_1+\cdots+m_n+n-1)!} = \int_{\Delta_1} x_1^{m_1} \cdots x_n^{m_n} dm(x).
\]
The exponential equality then follows from taking limits, since $\mathcal{P}_1$ and $\Delta_1$ are compact and  $e^{-\langle A, X \rangle_F}$ and $e^{-\langle \lambda, x \rangle}$ are limits of polynomials.

We note that this argument also implies the more general fact that $m$ is the pushforward of $\mu_1$ through the map $\phi: X \mapsto \diag(X)$:
\[
    \int_{\mathcal{P}_1} f(\phi(X)) d\mu_1(X) = \int_{\Delta_1} f(x) dm(x).
\]

\paragraph{From the simplex to a finite sum.} This transfer to the simplex now leads to an explicit computation for the integral when $A$ is a diagonal matrix  $A= \diag(\lambda)$ as a finite sum.
The assumption that $A$ is diagonal is without loss of generality because conjugating $A$ by any unitary does not change the integral $\int_{\mathcal{P}_1}e^{-\langle A,X\rangle_F} d\mu(X).$
By making a change of variables, the simplex integral becomes an iterated convolution:
\[
  \int_{\Delta_1} e^{-\langle \lambda, x \rangle} dm(x)   = (n-1)!\int_0^1 \int_0^{1-x_1} \cdots \int_0^{1-x_1-\cdots-x_{n-2}} e^{-\langle \lambda, x \rangle} dx = \left.(e^{-\lambda_1 t} * \cdots * e^{-\lambda_n t})\right|_{t=1}.\footnotemark{}
\]
This is stated formally in Lemma \ref{lem:laplace_integral}.
Applying the Laplace transform $\mathcal{L}$ converts this convolution into a partial fraction decomposition problem for distinct values of $\lambda_i$:
\footnotetext{Here, $*$ denotes the usual integral convolution: Given functions $f, g\colon \R\to \R$, $(f*g)(t)$ is $\int_{-\infty}^\infty f(t-\tau)g(\tau)d\tau$.}
\[
    \left.(e^{-\lambda_1 t} * \cdots * e^{-\lambda_n t})\right|_{t=1} = \mathcal{L}^{-1}\left[\frac{1}{\prod_i (s+\lambda_i)}\right](1) = \mathcal{L}^{-1}\left[\sum_i \frac{c_i}{s+\lambda_i}\right](1) = \sum_i c_i e^{-\lambda_i}.
\]
Computing the values of $c_i$ via a standard partial fractions formula gives:
\[
    \frac{1}{(n-1)!} \int_{\mathcal{P}_1} e^{-\langle Y, X \rangle_F} d\mu_1(X) = \sum_{i=1}^n \frac{e^{-\lambda_i}}{\prod_{j \neq i} (\lambda_j-\lambda_i)}.
\]
This is stated formally in Proposition \ref{prop:evaluation_formula}.

\paragraph{Proof of Theorem \ref{thm:partition}.}
We now state a lemma which gives the most basic result about integrals on $\mathcal{P}_1$.
Specifically, we state a well-known result for integrals of polynomial-like functions.
This proof is very related to the unitarily invariant inner product on homogeneous polynomials, which has many names in the literature: Bombieri-Weyl inner product, Fischer-Fock inner product, Segal-Bargmann inner product, etc.
The following lemma is standard; see e.g. Lemma 3.2 of \cite{pinasco2012}.

\begin{lemma}[Bombieri inner product formula] \label{lem:bombieri}
    For $\alpha \in \{0,1,2,\ldots\}^n$ such that $\sum_i \alpha_i = d$, we have:
    \[
        \int |v|^{2\alpha} d\mu_{S_\C^n}(v) = \int \prod_i |v_i|^{2\alpha_i} d\mu_{S_\C^n}(v) = \binom{d}{\alpha}^{-1} \binom{d+n-1}{n-1}^{-1} = \frac{\alpha_1! \cdots \alpha_n!(n-1)!}{(d+n-1)!}.
    \]
    Here, $\binom{d}{\alpha}$ is the multinomial coefficient and $\binom{d+n-1}{n-1}$ is the binomial coefficient.
\end{lemma}

\noindent
The next lemma relates the integrals we want to compute and the Laplace transform.

\begin{lemma} \label{lem:laplace_integral}
    For $\lambda_1,\ldots,\lambda_n \in \R$ and $x_n := 1-x_1-\cdots-x_{n-1}$, we have the following where $*$ denotes the usual integral convolution:
    \[
        \int_0^1 \int_0^{1-x_1} \cdots \int_0^{1-x_1-\cdots-x_{n-2}} e^{-\langle \lambda, x \rangle} dx_{n-1} \cdots dx_1 = \left.(e^{-\lambda_1 t} * \cdots * e^{-\lambda_n t})\right|_{t=1}.
    \]
    If $\lambda_1 < \lambda_2 < \cdots < \lambda_n$, then we further have:
    \[
        \int_0^1 \int_0^{1-x_1} \cdots \int_0^{1-x_1-\cdots-x_{n-2}} e^{-\langle \lambda, x \rangle} dx_{n-1} \cdots dx_1 = \sum_{i=1}^n \frac{e^{-\lambda_i}}{\prod_{j \neq i} (\lambda_j-\lambda_i)} 
    \]
\end{lemma}
\begin{proof}
    We first compute:
    \[
    \begin{split}
        \int_0^1 \int_0^{1-x_1} &\cdots \int_0^{1-x_1-\cdots-x_{n-2}} e^{-\langle  \lambda, x \rangle} dx_{n-1} \cdots dx_1 \\
            &= \int_0^1 e^{-\lambda_1x_1} \cdots \int_0^{1-x_1-\cdots-x_{n-2}} e^{-\lambda_{n-1}x_{n-1}} e^{-\lambda_n(1-x_1-\cdots-x_{n-1})} dx_{n-1} \cdots dx_1 \\
            &= \int_0^1 e^{-\lambda_1x_1} \cdots \int_0^{1-x_1-\cdots-x_{n-3}} e^{-\lambda_{n-2}x_{n-2}} \left.(e^{-\lambda_{n-1} t} * e^{-\lambda_n t})\right|_{t=1-x_1-\cdots-x_{n-2}} dx_{n-2} \cdots dx_1 \\
            &=\cdots = \int_0^1 e^{-\lambda_1x_1} \left.(e^{-\lambda_2 t} * \cdots * e^{-\lambda_n t})\right|_{t=1-x_1} dx_1 = \left.(e^{-\lambda_1 t} * \cdots * e^{-\lambda_n t})\right|_{t=1}.
    \end{split}
    \]
    Using the Laplace transform, we have $\mathcal{L}[e^{-\lambda_i t}](s) = \frac{1}{s+\lambda_i}$ which implies:
    \[
        e^{-\lambda_1 t} * \cdots * e^{-\lambda_n t} = \mathcal{L}^{-1}\left[\frac{1}{(s+\lambda_1)(s+\lambda_2) \cdots (s+\lambda_n)}\right](t).
    \]
    Assuming $\lambda_1 < \lambda_2 < \cdots < \lambda_n$, we can use Lagrange interpolation to compute:
    \[
        \mathcal{L}^{-1}\left[\frac{1}{(s+\lambda_1)(s+\lambda_2) \cdots (s+\lambda_n)}\right](t) = \mathcal{L}^{-1}\left[\frac{1}{Q'(-\lambda_1) \cdot (s+\lambda_1)} + \cdots + \frac{1}{Q'(-\lambda_n) \cdot (s+\lambda_n)}\right](t).
    \]
    Here, $Q(s) := (s+\lambda_1)(s+\lambda_2) \cdots (s+\lambda_n)$. With this we have:
    \[
        \mathcal{L}^{-1}\left[\frac{1}{Q'(-\lambda_1) \cdot (s+\lambda_1)} + \cdots + \frac{1}{Q'(-\lambda_n) \cdot (s+\lambda_n)}\right](t) = \frac{e^{-\lambda_1 t}}{Q'(-\lambda_1)} + \cdots + \frac{e^{-\lambda_n t}}{Q'(-\lambda_n)} = \sum_{i=1}^n \frac{e^{-\lambda_i t}}{\prod_{j \neq i} (\lambda_j-\lambda_i)}.
    \]
    Plugging in $t=1$ gives the equality in the second statement. 
\end{proof}

\noindent
We now state and prove the full evaluation formula for $\mathcal{P}_1$ in full generality.

\begin{proposition}[Generalization of Theorem \ref{thm:partition}] \label{prop:evaluation_formula}
    Fix $n \in \N$, and let $\mu_{S_\C^n}, \mu_1,$ and $\mu_{\Delta_1}$ be the uniform probability distributions on the complex unit sphere in $\C^n$, on $\mathcal{P}_1$, and on the standard simplex in $\R^n$, respectively. For a given analytic function $f$ on the standard simplex the following expressions are equal:
    \begin{enumerate}
   \item $\displaystyle \int_{S_\C^n} f(|v_1|^2,\ldots,|v_n|^2) d\mu_{S_\C^n}(v)$,
        \item $\displaystyle \int_{\mathcal{P}_1} f(\diag(X)) d\mu_1(X)$,
        \item $\displaystyle \int_{\Delta_1} f(x) d\mu_{\Delta_1}(x)$,
        \item $\displaystyle (n-1)! \int_0^1 \int_0^{1-x_1} \cdots \int_0^{1-x_1-\cdots-x_{n-2}} f(x_1,\ldots,x_{n-1},1-x_1-\cdots-x_{n-1}) dx_{n-1} \cdots dx_1$.
    \end{enumerate}
\end{proposition}
\begin{proof}
    First, for the equality of $(1)$ and $(2)$, note that $\mu_1$ is the pushforward measure of $\mu_{S_\C^n}$ through the map $\psi: S_\C^n \to \mathcal{P}_1$ given by $\psi: v \mapsto vv^*$. (To see this, note that $\psi$ is unitarily invariant and $\mu_{S_\C^n}$ and $\mu_1$ are the unique unitarily invariant measures on their domains.) With this, we have:
    \[
        \int_{\mathcal{P}_1} f(\diag(X)) d\mu_1(X) = \int_{S_\C^n} f(\diag(\psi(v))) d_{S_\C^n}(v) = \int_{S_\C^n} f(|v_1|^2,\ldots,|v_n|^2) d_{S_\C^n}(v).
    \]
    That is, $(1)$ and $(2)$ are equal.

    Next, the equality of $(3)$ and $(4)$ follows from the fact that the map between the two domains of integration (both of which are simplices) is affine. Therefore the determinant of the Jacobian is a constant, and thus we only need to integrate over a constant function to determine that constant. A simple induction shows that it is $(n-1)!$.
    
    To prove the equality of $(1)$ and $(4)$, we compute the integrals on a given monomial $x^m := x_1^{m_1} \cdots x_{n-1}^{m_{n-1}} (1-x_1-\cdots-x_{n-1})^{m_n}$. First, by Lemma \ref{lem:bombieri} we have:
    \[
        \int |v|^{2m} d\mu_{S_\C^n}(v) = \frac{m_1! \cdots m_n! (n-1)!}{(|m|+n-1)!} = \binom{|m|+n-1}{m_1,\ldots,m_n,n-1}^{-1}.
    \]
    Now, note:
    \[
        \int_0^1 x^l (1-x)^m dx = \sum_{k=0}^m \binom{m}{k} (-1)^k \int_0^1 x^{l+k} dx = \frac{l!m!}{(l+m+1)!} \sum_{k=0}^m \frac{\prod_{j \neq k} (l+j+1)}{(-1)^k k! (m-k)!} = \frac{l!m!}{(l+m+1)!}.
    \]
    The last equality is due to Lagrange interpolation, considering the sum as a function of $n$. Further:
    \[
        \frac{l!m!}{(l+m+1)!} = \binom{l+m+2-1}{l,m,2-1}^{-1}.
    \]
    That is, we have equality whenever $n=2$, proving the base case. The rest of the proof goes by induction. First we compute for $\alpha = 1-x_1-\cdots-x_{n-2}$:
    \[
        \int_0^\alpha x_{n-1}^j (\alpha-x_{n-1})^k dx_{n-1} = \int_0^1 (\alpha u)^j (\alpha - \alpha u)^k \alpha du = \alpha^{j+k+1} \int_0^1 u^j (1-u)^k du = \frac{\alpha^{j+k+1} \cdot j!k!}{(j+k+1)!}.
    \]
    With this, we then compute the following by induction, letting $\beta = 1-x_1-\cdots-x_{n-3}$:
    \[
    \begin{split}
        (n-1)! \int_0^1 &\cdots \int_0^\alpha x_1^{m_1} \cdots x_{n-1}^{m_{n-1}} (\alpha-x_{n-1})^{m_n} dx_{n-1} \cdots dx_1 \\
            &= \frac{m_{n-1}!m_n!}{(m_{n-1}+m_n+1)!} \cdot (n-1)! \int_0^1 \cdots \int_0^{\alpha+x_{n-2}} x_1^{m_1} \cdots x_{n-2}^{m_{n-2}} \alpha^{m_{n-1}+m_n+1} dx_{n-2} \cdots dx_1 \\
            &= \frac{m_{n-1}!m_n!(n-1)}{(m_{n-1}+m_n+1)!} \left[(n-2)! \int_0^1 \cdots \int_0^\beta x_1^{m_1} \cdots x_{n-2}^{m_{n-2}} (\beta-x_{n-2})^{m_{n-1}+m_n+1} dx_{n-2} \cdots dx_1\right] \\
            &= \frac{m_{n-1}!m_n!(n-1)}{(m_{n-1}+m_n+1)!} \cdot \frac{m_1! \cdots m_{n-2}! (m_{n-1}+m_n+1)! (n-1-1)!}{(m_1+\cdots+m_n+1+n-1-1)!} \\
            &= \frac{m_1! \cdots m_n! (n-1)!}{(m_1+\cdots+m_n+n-1)!}.
    \end{split}
    \]
    This completes the proof of equality of $(1)$ and $(4)$.
    
\end{proof}

\section{Groups, Manifolds, and a Brief History of Lie Theory} \label{sec:symmetries_revisited}
The two examples of the previous sections make it clear that understanding the symmetries of a particular problem are crucial to understanding how that problem is solved. 
The optimization and sampling problems on the unitary group posed earlier were able to be simplified by reducing down to the discrete symmetric group. 
In this and the next sections, we  observe and discuss the deeper connections between such continuous and discrete groups for the special class of Lie groups. 
By understanding these connections, we hope to convince the reader that the structure of the example problems discussed above goes much deeper than those examples might initially suggest.

Before moving on, we first need to define formally the fundamental objects with which we will be working: groups and manifolds. 
A \textbf{group} $G$ is a set with a distinguished element $e \in G$ called the \textbf{identity}, equipped with a product operation $*$ and an inversion operation $\cdot^{-1}$. 
The product operation should have the property that $g * e = e * g = g$ for all $g \in G$, the inversion operation should have the property that $e^{-1} = e$, and together they should satisfy $g * g^{-1} = e$. 
Typical examples of groups are the symmetric group $S_n$ and the general linear group $\GL_n(\R)$ of invertible $n \times n$ real matrices.

Informally, a \textbf{(differentiable) manifold} $M$ is a topological space which locally looks like $\R^n$ for some $n$ at every point. 
Formally, this means that in some open neighborhood $U_x$ about any point $x \in M$ there is a homeomorphism $\phi_x: U_x \to V_x \subset \R^n$, such that if $U_x \cap U_y \neq \varnothing$ then $\phi_x \circ \phi_y^{-1}: V_y \to V_x$ is a smooth map (see \cite{Lee00,VGeodesic} for more on manifolds). {Note that one can also replace $\R^n$ by $\C^n$ and ``smooth'' by ``holomorphic'' in this definition to obtain the notion of a \textbf{complex manifold}.} 
 
Given a group $G$ and a manifold $M$, the \textbf{action} of $G$ on $M$ is a function $\phi: G \times M \to M$ for which we write $(g,x)$ as $g \cdot x$. 
A \textbf{group action} must further satisfy the properties that $g \cdot (h \cdot x) = (gh) \cdot x$ and $e \cdot x = x$. 
An \textbf{orbit} of $x \in M$ is the set $\mathcal{O}_x := \{y \in M : y = g \cdot x \text{ for some } g \in G\}$. The \textbf{stabilizer} of $x \in M$ is the set $\Stab_x := \{g \in G : g \cdot x = x\}$. 
A \textbf{group representation} is a vector space $V$ that is acted upon by a group action of $G$. Another way to think of this is to consider a map $\rho$ from an element of the group $g \in G$ to an invertible linear transformation $\rho(g)$ of the vector space $V$. Enforcing $\rho$ to be a group homomorphism (i.e., for any $g_1,g_2 \in G$ we have $\rho(g_1 g_2)=\rho(g_1)\rho(g_2)$) implies the action $g \cdot v := \rho(g) v$ is a group action.

The symmetries of an object are given by a group acting on that object. 
By finding properties that are preserved by these symmetries, we are finding coordinate-free and canonical properties of the underlying object. 
This allows us to get at the heart of the object without having to deal with the technicalities that often come with a particular choice of coordinates or parameterization of the object.

Let us  consider the simple example of the vector space $\R^n$. 
This vector space can be considered to have two continuous groups that act upon it, given by translations and rotations. 
For translation, we have that the additive group $(\R^n,+)$ acts on $\R^n$ via $(x_1,\ldots,x_n) \mapsto (x_1+c_1,\ldots,x_n+c_n)$ for any $c$ in the group $(\R^n,+)$. 
This action is abelian or commutative, since changing the order of mutliple translations does not change the effect of the action. 
For rotation, the group $\SO_n(\R)$ of orthogonal rotation matrices acts on $\R^n$ via $x \mapsto Ox$ for any $O \in \SO_n(\R)$. 
This action is \emph{not} abelian; different orderings of rotations of $\R^n$ can result in different actions.

Given a subset $S \subset \R^n$, the translation and rotation actions preserve many important properties that one might want to know about $S$: volume, surface area, width, distances, and angles. 
The translation action also preserves some properties that rotation does not, for example the slope of a line. 
The rotation action similarly preserves other properties as well, like when a set $S$ is centrally symmetric. 
The key point  is that these symmetries bring to the forefront various properties that we find important or interesting about a set $S$. 
The symmetries serve to remove the coincidental data regarding how a particular set is embedded in $\R^n$, leaving behind only the important information about the set.

To see how such continuous symmetries play a role in the development of algorithms, we now turn to Lie theory. 
As discussed above, Lie theory is be used to generalize and shed light on the two example problems explored in the previous sections. 
And before discussing the basics of Lie theory in general, we first discuss a number of important properties of the unitary group $\U(n)$. 
The unitary group $\U(n)$ is essentially the complex vector space equivalent of the group $\SO(n)$ of rotations of a real vector space. 
The introduction of the complex numbers here serves to simplify the theory, and this is a common theme in Lie theory. 
This example of the unitary group $\U(n)$  serves as a running example for the remainder of our discussion.

Before moving on, we mention a bit of the history of Lie theory. 
Lie theory is a vast subject that has been studied for around 150 years by some of the greatest mathematicians of that time period.
It has had a close association with physics starting with the work of Emmy Noether \cite{Noether1918} and, till today, plays an important role in the search of fundamental laws of nature.
The subject is named after Sophus Lie, who was formally initiated the study of  infinitesimal group actions on a manifold.
This gave rise to the Lie algebra of a Lie group, and Wilhelm Killing then extensively studied the problem of the classification of Lie algebras \cite{killing1889zusammensetzung}. 
Killing gave the correct classification, albeit with a few incomplete or incorrect proofs, and Elie Cartan gave a rigorous proof based on Killing's work (e.g., see \cite{chern1952elie}). 
For a more detailed account of the history of this classification and the emergence of Lie theory, most of which occurred in the 19th century, see \cite{hawkins2012emergence}. 
Later works of E. B. Dynkin \cite{dynkin1947structure} and Nathan Jacobson \cite{jacobson1979lie} made this classification and the surrounding results more accessible, and their results and exposition often serve as the basis of how Lie theory is taught today.

While a classification of Lie algebras gave way to a number of results, mathematicians of the 20th century continued searching for more underlying structure which would enable proofs that did not rely on the classification. 
To this end, Hermann Weyl (see \cite{weyl1968gesammelte}), Claude Chevalley \cite{chevalley1948theorie}, Harish-Chandra \cite{harish1951some}, Hendrik Casimir and Bartel Leendert van der Waerden \cite{casimir1935algebraischer}, and Jean-Pierre Serre (see \cite{serre2012complex}) all played significant roles in the development of this theory, among countless others. 
Many of their findings are the underpinnings of the Lie theory that we  present here.

\section{The Unitary Group} \label{sec:unitary_group}

Given $n \in \N$, the \textbf{unitary group} $\U(n)$ is the group of $n \times n$ unitary matrices under matrix multiplication. 
This is defined more explicitly by
\[
    \U(n) := \{U \in \C^{n \times n} : UU^* = I\},
\]
where $\C^{n \times n}$ is the vector space of all complex $n \times n$ matrices, $U^*$ denotes the conjugate transpose of $U$, and $I$ is the identity matrix. 
Because this group in embedded in the vector space of matrices, it can be equipped with a real\footnote{Even though the unitary group is embedded in the complex vector space $\C^{n \times n}$, it cannot be equipped with a complex manifold structure. 
However, it can be equipped with a real manifold structure by viewing $\C^{n \times n} \cong (\R^{n \times n})^2$. 
This is similar to the fact that the space of Hermitian matrices is a real vector space, even though the matrices have complex entries.} manifold structure coming from the vector space. 
As a manifold, $\U(n)$ is compact. 
Further, the group operations of multiplication and inversion in $\U(n)$, given by $V \mapsto UV$ for fixed $U \in \U(n)$ and $V \mapsto V^{-1}$, are smooth maps with respect to the manifold structure. 
That is, the algebraic group structure and the analytic manifold structure are compatible.

As a manifold, $\U(n)$ has a tangent space over every point $U \in \U(n)$. 
These tangent spaces are isomorphic as vector spaces, so let us in particular look at the tangent space over a canonical point, the identity matrix $I \in \U(n)$. 
We compute the elements of this tangent space in a standard way: by considering derivatives of paths on the manifold through the identity matrix. 
That is, fix any smooth function $f: (-\epsilon, \epsilon) \to \U(n)$ such that $(-\epsilon, \epsilon) \subset \R$ is a small real interval and $f(0) = I$. 
An element $X$ of the tangent space of $\U(n)$ at the identity, denoted $T_I \U(n)$, is given by
\[
    X = \left.\frac{d}{dt}\right|_{t=0} f(t),
\]
and in fact every element of the tangent space of $\U(n)$ can be constructed in this way. 
So far this construction has only used the manifold structure of $\U(n)$, but we now employ the group structure to get a better handle on what $X$ can actually be. 
Since $f(t) \in \U(n)$ for all $t$ near 0, we further have that $f(t) f(t)^* = I$ for all such $t$. 
By the derivative Leibniz rule, this implies
\[
    0 = \left.\frac{d}{dt}\right|_{t=0} \left[f(t) f(t)^*\right] = \left[\left.\frac{d}{dt}\right|_{t=0} f(t)\right] f(0)^* + f(0) \left[\left.\frac{d}{dt}\right|_{t=0} f(t)\right]^* = X + X^*.
\]
In fact, this precisely describes the tangent space $T_I \U(n)$:
\[
    T_I \U(n) = \{X \in \C^{n \times n} : X^* = -X\}.
\]
That is, $T_I \U(n)$ is the (real) vector space of all $n \times n$ skew-Hermitian matrices.

Beyond describing the tangent space, the group structure on $\U(n)$ also provides extra structure on the tangent space $T_I \U(n)$. 
For example, the conjugation action of $\U(n)$ on $\C^{n \times n}$ preserves the space of skew-Hermitian matrices. 
Specifically, for any $U \in \U(n)$ and any $X \in T_I \U(n)$ we have
\[
    (UXU^{-1})^* = U^{-*}X^*U^* = UX^*U^{-1} = -UXU^*,
\]
which precisely says that $UXU^{-1}$ is skew-Hermitian. 
Another interesting preserving ``action'' of $T_I \U(n)$ on itself is given by the matrix commutator, $[X,Y] := XY-YX$. 
For any $X,Y \in T_ \U(n)$, we have
\[
    (XY-YX)^* = Y^*X^*-X^*Y^* = (-Y)(-X)-(-X)(-Y) = -(XY-YX),
\]
which precisely says that $XY-YX$ is skew-Hermitian. 
And finally, a general vector space can be equipped with an inner product in any number of ways, but the tangent space $T_I \U(n)$ has a very special inner product $\langle X, Y \rangle := -\Tr(XY)$. (Note that the negation is required for positive definiteness, since $X,Y$ are \emph{skew}-Hermitian.) 
What makes this inner product so special is its connection to the above actions we described.
Specifically,
\[
    \langle UXU^{-1}, UYU^{-1} \rangle = -\Tr(UXU^{-1}UYU^{-1}) = -\Tr(UXYU^{-1}) = -\Tr(XY) = \langle X, Y \rangle,
\]
and
\[
\begin{split}
    \langle [X,Y], Z \rangle &= -\Tr((XY-YX)Z) = -\Tr(XYZ-YXZ) = \Tr(YXZ-XYZ) \\
        &= \Tr(YXZ-YZX) = \Tr(Y(XZ-ZX)) = -\langle Y, [X,Z] \rangle.
\end{split}
\]
That is, this inner product is \textbf{invariant} under these natural actions of conjugation and commutation.

With this, we  catch a glimpse of the symmetric nature of the unitary group $\U(n)$. 
The above discussion has made it clear that the tangent space of $\U(n)$ is more than just a vector space: The unitary group acts on it in a natural way, it acts on itself in a natural way, and these actions are very compatible with a natural inner product. 
The obvious next questions are as follows: 
What can we gain from all this structure, in terms of optimization and sampling? And how specific to the unitary group is all of this structure?

\section{Lie Theory Basics}\label{sec:Lie}

The structure of the unitary group and its tangent space is not specific to the unitary group, but in fact generalizes to \textbf{Lie groups} and their associated \textbf{Lie algebras}. 
While the features of the unitary group discussed above may have seemed somewhat ad hoc, we  see how they can be derived from general principles in the context of Lie theory. 
After discussing the basics of Lie theory, we  then utilize this theory and demonstrate its importance in the context of optimization and sampling. 
All of this basic material is standard, and can be found in standard references, e.g., \cite{baker2003matrix,knapp2013,hall2003lie}.
Chapter 4 of \cite{knapp2013} is a particularly good reference for compact Lie groups. 
For a view more towards physics and engineering, see also \cite{gilmore2008lie}.

As a final note, while we discuss a number of basic Lie theoretic topics, we do not discuss more advanced topics which would be considered to be crucial to a complete discussion of Lie theory. This includes Lie group and Lie algebra representations, roots, weight theory, Dynkin diagrams, nilpotent Lie algebras, the exponential map, and likely other topics.

\subsection{Lie Groups}

A \textbf{Lie group} $G$ is a group equipped with a manifold structure which is compatible with the group structure.
Compatible here means that the operations of multiplication $g*h$ and inversion $g^{-1}$, given by
\[
    *: G \times G \to G \qquad \text{and} \qquad \cdot^{-1}: G \to G,
\]
are smooth (differentiable) maps between manifolds. 
The typical examples of Lie groups are \textbf{matrix groups}, which are groups of matrices equipped with the manifold structure coming from the fact that a matrix can be considered as a point in $\C^{n \times n}$. 
Some examples are:
\begin{itemize}
    \item $\U(n)$, the group of $n \times n$ unitary matrices.
    \item $\SU(n)$, the group of $n \times n$ unitary matrices with determinant 1.
    \item $\GL_n(\C)$ and $\GL_n(\R)$, the groups of invertible $n \times n$ matrices.
    \item $\GL(V)$, the group of invertible linear maps from the vector space $V$ to itself.
    \item $\SL_n(\C)$ and $\SL_n(\R)$, the groups of invertible $n \times n$ matrices with determinant 1.
    \item $\SO_n(\C)$ and $\SO_n(\R)$, the groups of orthogonal $n \times n$ matrices with determinant 1.
    \item $\Sp_{2n}(\C)$, the group of $2n \times 2n$ symplectic matrices.
    \item $\USp(n)$, the group of $2n \times 2n$ symplectic unitary matrices.
\end{itemize}
Note that even though we often consider matrix groups contained in $\C^{n \times n}$, certain matrix groups may still be considered as either real Lie groups or complex Lie groups. 
In the case of a complex Lie group, $G$ is a complex manifold and the multiplication and inversion maps must actually be holomorphic. 
In the case of a real Lie group, we implicitly identify $\C^{n \times n} \cong (\R^{n \times n})^2$. 
Note that for some matrix groups $G$, like the unitary group $G = \U(n)$, we are forced to consider $G$ as a real Lie group even though elements of the group have complex entries. 
This situation occurs when there is no way to endow our Lie group $G$ with the structure of a complex manifold, and one important example of such groups are the compact Lie groups. 
This point is not very important to our exposition, but it is still worth making to avoid certain confusions.

Given a Lie group $G$, one can define the notion of a \textbf{Lie subgroup}. 
A Lie subgroup $H$ of $G$ is a subgroup of $G$ which is embedded within $G$ as a submanifold. 
A theorem of Cartan \cite{cartan1952theorie} says that if $H$ is a closed subgroup of a real Lie group $G$, then $H$ is in fact a Lie subgroup of $G$. 
A wealth of examples of Lie subgroups come from the matrix Lie groups listed above. 
All of them are Lie subgroups of either $\GL_n(\C)$ or $\GL_n(\R)$. 
Further, $\SU(n)$ is a Lie subgroup of $\U(n)$.
Another important example is the \textbf{compact torus} $\U(1)^n$, where $\U(1)$ is isomorphic to the unit circle. 
There are a number of ways to view $\U(1)^n$ as a Lie subgroup of $\U(n)$, but the typical way of embedding of $\U(1)^n$ into $\U(n)$ is as the set of all diagonal matrices in $\U(n)$.

\subsection{The Tangent Space at the Identity Element}

Since a Lie group $G$ is a manifold, it has a tangent vector space at every point. 
Moreover, since these tangent spaces are all isomorphic as vector spaces, we restrict to a specific canonical tangent space; that is, we restrict to the tangent space of the identity element of the group, denoted $T_e G$. 
Using only the manifold structure of a Lie group $G$, there is a standard way to characterize the tangent space of $G$ at a given point, and we now describe this for $T_e G$. 
This is precisely what we did for the group $\U(n)$ above. 
Let $f: (-\epsilon,\epsilon) \to G$ be an injective smooth map from a small real interval about 0 to the Lie group $G$ such that $f(0) = e$, the identity element. 
Such an $f$ is equivalently defined as a short path in $G$ which passes through the identity element at $t=0$. The expression
\[
    X = \left.\frac{d}{dt}\right|_{t=0} f(t)
\]
then gives rise to an element $X$ of the tangent space $T_e G$ according to the direction of the path at the identity. 
Further, every element of the tangent space can be formed in this way, by choosing an appropriate smooth path through the identity element of $G$.

\begin{example}[$\GL_n(\C)$] \label{ex:GL_n}
    Beyond the unitary group $\U(n)$ above, the simplest example is $\GL_n(\C)$ (and the same argument works for $\GL_n(\R)$). 
    Since this Lie group is embedded in the space of $n \times n$ complex matrices $\C^{n \times n}$, the elements of the tangent space $T_I \GL_n(\C)$ are  $n \times n$ complex matrices. 
    Since the determinant of a matrix is continuous in its entries, it is straightforward to see that $\GL_n(\C)$ is an $n^2$-dimensional complex manifold. 
    Thus the tangent space $T_I \GL_n(\C)$ is an $n^2$-dimensional subspace of the vector space of all complex $n \times n$ matrices, and therefore $T_I \GL_n(\C)$ must in fact be equal to the whole vector space of such matrices.
\end{example}

\subsection{The (Group) Adjoint Action}

We now incorporate the group structure of $G$ by defining an action of the group $G$ on the tangent space $T_e G$. 
Given $X \in T_e G$, let $f: (-\epsilon, \epsilon) \to G$ be a short path in $G$ such that $f(0) = e$ and $\left.\frac{d}{dt}\right|_{t=0} f(t) = X$. We now act on $X$ by $g \in G$ via the \textbf{(group) adjoint action}, given by
\[
    \Ad_g(X) := \left.\frac{d}{dt}\right|_{t=0} \left[g \cdot f(t) \cdot g^{-1}\right].
\]
Note that by the compatibility properties of the group and manifold structures of $G$, we have that $g \cdot f(t) \cdot g^{-1}$ is a smooth path such that $\left[g \cdot f(t) \cdot g^{-1}\right](0) = g \cdot e \cdot g^{-1} = e$.
Therefore the above expression for $\Ad_g(X)$ gives a well-defined element of the tangent space $T_e G$. 
It is then technical but straightforward that this action does not depend on the choice of $f$, and in turn that this action gives rise to a well-defined group action on $T_e G$.

In the case of matrix Lie groups, such as those listed above, there is a straightforward way to define this (group) adjoint action. 
This is given by matrix conjugation, as described in the case of the unitary group in Section \ref{sec:unitary_group}. 
Specifically, if $G$ is a matrix group, then the tangent space $T_e G$ can be considered as a vector space of matrices. The (group) adjoint action of $G$ on $T_e G$ is given by
\[
    \Ad_g: X \mapsto gXg^{-1}
\]
for all $g \in G$ and $X \in T_e G$. 
For essentially everything we discuss here, the reader can keep this conjugation action in mind for the (group) adjoint action of a Lie group $G$ on its tangent space $T_e G$.

\begin{remark}[Adjoint orbits] \label{rem:adjoint_orbits_KKS}
    The (group) adjoint action of $G$ on its tangent space $T_e G$ is a group action, and thus this action partitions the vector space $T_e G$ into orbits, referred to as the \textbf{adjoint orbits} of $G$. 
    Given $X \in T_e G$, we denote the adjoint orbit of $X$ by $\mathcal{O}_X := \{\Ad_g(X) : g \in G\}$. In the case of the unitary group, the adjoint orbits are determined precisely by the eigenvalues of the matrices in $T_I \U(n)$. That is, the adjoint orbits of $\U(n)$ can be defined via
    \[
        \mathcal{O}_X =\{Y \in T_I \U(n) : \eig(Y) = \eig(X)\}.
    \]
    For compact groups $G$, we note that such adjoint orbits can be given the structure of a symplectic manifold via the \textbf{Kirillov-Kostant-Souriau (KKS) symplectic form}. 
    We do not discuss this further here, but instead direct the interested reader to \cite{kirillov2004lectures}. 
    We also note that this KKS form typically refers to \emph{co}adjoint orbits rather than adjoint orbits (as in \cite{kirillov2004lectures}), but these are essentially equivalent in our case, see Remark \ref{rem:coadjoint}.
\end{remark}

\subsection{The (Algebra) Adjoint Action} \label{sec:algebra_adjoint_action}

When considering the unitary group, we constructed a certain action of the tangent space on itself, given by the matrix commutator.
It is easy to see that this operation is bilinear, and as we saw in Section \ref{sec:unitary_group}, it preserves the tangent space $T_I \U(n)$ of the unitary group. Because of this, the operation of matrix conjugation can be thought of as a product-like operation on $T_I \U(n)$.

In a general Lie group, this action, called the \textbf{(algebra) adjoint action}, can be constructed from the group adjoint action discussed in the previous section. 
To do this, we take the derivative of the action of $\Ad_g$. 
More formally, the action of $G$ on $T_e G$ can be thought of as a group homomorphism, given by
\[
    \phi: G \to \GL(T_e G) \qquad \text{via} \qquad \phi: g \mapsto \Ad_g.
\]
The group $\GL(T_e G)$ is in fact a Lie group, and thus a manifold, and the above map $\phi$ is a smooth group homomorphism. 
Thus we can consider the differential of the map $\phi$. By definition, this differential gives rise to linear maps from the tangent spaces of $G$ to the tangent spaces of $\GL(T_e G)$. 
When restricted to the tangent space $T_e G$, this gives rise to a map
\[
    d\phi_e: T_e G \to T_I \left[\GL(T_e G)\right] \qquad \text{via} \qquad d\phi_e: X \mapsto \ad_X.
\]
That is, we have defined $\ad_X$ to be the element of $T_I \left[\GL(T_e G)\right]$ given by $d\phi_e(X)$ for $X \in T_e G$. 
It not yet clear precisely what this $\ad_X$ is, and we  describe this further now.

We have already described $T_I \left[\GL(T_e G)\right]$ as the vector space of all (not necessarily invertible) linear maps from $T_e G$ to itself, see Example \ref{ex:GL_n}. 
Thus $\ad_X$ is a linear map on the tangent space $T_e G$ of $G$. We then refer to this as the \textbf{(algebra) adjoint action} of $T_e G$ on itself.

In the case that $G$ is a matrix group, we now compute explicitly this (algebra) adjoint action. 
Given a short path $f: (-\epsilon,\epsilon) \to G$ such that $f(0) = I$ and $\left.\frac{d}{dt}\right|_{t=0} f(t) = X$, and any $Y \in T_I G$, the differential is defined via
\[
\begin{split}
    d\phi_I\left(\left.\frac{d}{dt}\right|_{t=0} f(t)\right) Y &= \left.\frac{d}{dt}\right|_{t=0} \phi(f(t)) Y = \left.\frac{d}{dt}\right|_{t=0} \left[f(t) Y f(t)^{-1}\right] \\
        &= \left(\left.\frac{d}{dt}\right|_{t=0} f(t)\right) Y f(0)^{-1} + f(0) Y \left(\left.\frac{d}{dt}\right|_{t=0} f(t)^{-1}\right) \\
        &= XY + Y\left(-f(0)^{-2} \left.\frac{d}{dt}\right|_{t=0} f(t)\right) = XY - YX.
\end{split}
\]
That is, $\ad_X(Y) = XY-YX$ is the matrix commutator for all matrix groups $G$, including the unitary group.

\subsection{Lie Algebras} \label{sec:lie_alg_via_group}

With the algebra adjoint action in hand, we are ready to define the Lie algebra of a Lie group. 
Given a Lie group $G$, its \textbf{Lie algebra}, denoted $\mathfrak{g}$, is the tangent vector space $T_e G$ equipped with a product-like operation called the \textbf{Lie bracket}. 
The Lie bracket on $\mathfrak{g}$ is defined via
\[
    [\cdot, \cdot]: \mathfrak{g} \times \mathfrak{g} \to \mathfrak{g} \qquad \text{via} \qquad [X,Y] := \ad_X(Y).
\]
The Lie bracket has a few important properties:
\begin{enumerate}
    \item \textbf{Bilinearity:} $[aX+bY,Z] = a[X,Z] + b[Y,Z]$ and $[X,bY+cZ] = b[X,Y] = c[X,Z]$.
    \item \textbf{Anti-symmetry:} $[X,Y] = -[Y,X]$.
    \item \textbf{Jacobi identity:} $[X,[Y,Z]] + [Y,[Z,X]] + [Z,[X,Y]] = 0$.
\end{enumerate}
While the Lie bracket is referred to as a product-like operation, it is not a typical product because it is not associative. 
However, one can rewrite the Jacobi identity to view it as an ``associativity error'' formula as follows:
\[
    [[X,Y],Z] - [X,[Y,Z]] = [Y,[Z,X]].
\]
In the case of Lie algebras associated to matrix groups, the Lie bracket is given by the commutator $[X,Y] = XY-YX$ as discussed above.
It is straightforward to see that the matrix commutator satisfies these three properties.

As with Lie subgroups of Lie groups, there is a notion of a \textbf{Lie subalgebra} of a Lie algebra. A Lie subalgebra $\mathfrak{h} \subset \mathfrak{g}$ is a vector subspace with Lie bracket inherited from $\mathfrak{g}$ such that $[X,Y] \in \mathfrak{h}$ for all $X,Y \in \mathfrak{h}$. 
According to the Lie group-Lie algebra correspondence, every such Lie subalgebra $\mathfrak{h} \subset \mathfrak{g}$ is the Lie algebra of a Lie subgroup $H \subset G$.

There are various categories that different Lie algebras can be put into, but the most important ones for us are \textbf{compact Lie algebra}, \textbf{abelian Lie algebra}, and \textbf{semisimple Lie algebra}. 
A compact Lie algebra is a Lie algebra associated to a compact Lie group. 
For example, the Lie algebra $\mathfrak{u}_n$ associated to the Lie group $\U(n)$ is a compact Lie algebra. 
An abelian Lie algebra is a Lie algebra with trivial Lie bracket: $[X,Y] = 0$ for all $X,Y \in \mathfrak{g}$. 
A semisimple Lie algebra is a Lie algebra with trivial \textbf{center}; i.e., for which there is no nonzero $X \in \mathfrak{g}$ such that $[X,Y] = 0$ for all $Y \in \mathfrak{g}$. Note that $\mathfrak{su}_n$ is semisimple, but $\mathfrak{u}_n$ is not because the matrix $i \cdot I$ is skew-Hermitian and commutes with all matrices in $\mathfrak{u}_n$.

\begin{example}[Lie algebra of the unitary group]
    In the case of the unitary group $\U(n)$, the associated Lie algebra $\mathfrak{u}_n$ is the (real) vector space of all skew-Hermitian matrices, as discussed in Section \ref{sec:unitary_group}. As stated above, the Lie bracket of $\mathfrak{u}_n$ is the commutator since $\U(n)$ is a matrix group. To see that this is well-defined, we compute the following for any skew-Hermitian matrices $X,Y$:
    \[
        [X,Y]^* = (XY-YX)^* = Y^*X^* - X^*Y^* = -XY + YX = -[X,Y].
    \]
    That is, we have that $[X,Y]$ is itself skew-Hermitian, and thus $[\cdot,\cdot]: \mathfrak{u}_n \times \mathfrak{u}_n \to \mathfrak{u}_n$ is well-defined.
\end{example}

\begin{example}[Lie algebra of the compact torus]
    The subspace of diagonal matrices in $\mathfrak{u}_n$, all of which have purely imaginary diagonal entries, is an abelian Lie subalgebra of the compact Lie algebra $\mathfrak{u}_n$. This Lie subalgebra is then the Lie algebra of the compact abelian subgroup $\U(1)^n \subset \U(n)$, given by the diagonal matrices in $\U(n)$.
\end{example}
Although we have constructed Lie algebras here via the tangent spaces of Lie groups, one can also define an \textbf{abstract Lie algebra} as any vector space with a Lie bracket operation which satisfies the above three properties. 
Such a Lie algebra is not a priori associated to a Lie group, but a classical classification theorem shows that nice Lie algebras defined in this way can be associated to matrix Lie groups.
We discuss this classification further in Section \ref{sec:classification}.

\section{Lie Theory for Optimizing Over Orbits}

Having discussed the basics of Lie theory, we were able to answer one of our questions; i.e., the structure we found in the unitary group and its tangent space is not at all specific to the unitary group. 
Given a Lie group $G$ in general, one can construct the associated Lie algebra $\mathfrak{g}$ using the tangent space at the identity, and the conjugation action and commutator of the unitary group $\U(n)$ and its Lie algebra $\mathfrak{u}_n$ generalize to the adjoint actions $\Ad_g$ and $\ad_X$ and the Lie bracket $[\cdot, \cdot]$ of any $G$ and $\mathfrak{g}$.

This leaves the other question: What can we gain from all this structure, in terms of optimization and sampling? 
To answer this question, we first discuss a  simple optimization problem that one can associate to the adjoint orbits of a Lie group $G$. 
To get a sense of this optimization problem, we first consider a related problem for the unitary group.

Let $\mathcal{H}(n)$ denote the (real) vector space of all $n \times n$ Hermitian matrices, let $Z \in \mathcal{H}(n)$ be any real diagonal matrix, let $X = e_1 e_1^\top \in \mathcal{H}(n)$, and let $\mathcal{O}_X$ denote the conjugation orbit of $X$ with respect to the group $\U(n)$. We consider the optimization problem
\[
    \min_{Y \in \mathcal{O}_X} \langle Y, Z \rangle_F,
\]
where $\langle Y, Z \rangle_F = \Tr(YZ)$ is the Frobenius inner product on Hermitian matrices.
Another way to write this is as
\[
    \min_{Y \in \mathcal{O}_X} \langle Y, Z \rangle_F = \min_{\|v\|_2 = 1} v^* Z v.
\]
It is well-known that this problem is equivalent to computing the minimum eigenvalue of $Z$, which is precisely the problem studied in Section \ref{sec:eigenvalue} above.
Since $Z$ is diagonal, this problem is further equivalent to computing the minimum entry of the diagonal of $Z$. 
This can be written as
\[
    \min_{Y \in \mathcal{O}_X} \langle Y, Z \rangle_F = \min_{\sigma \in S_n} (\sigma \cdot e_1)^\top Z (\sigma \cdot e_1),
\]
where $S_n$ is the symmetric group and $\sigma \cdot v$ acts on the vector $v$ by permuting its entries.

The next question then is: Why study further such a simple optimization problem? 
What we see in the rest of this section is that optimization problem, while simple, has a Lie-theoretic interpretation which demonstrates an interesting connection between the nonconvex adjoint orbits of a Lie group and certain associated convex polytopes. 
To do this, we need to generalize some of the notions used in the above simple optimization problem: specifically, the Frobenius inner product, the notion of a diagonal matrix, and the symmetric group $S_n$.

\subsection{The Killing Form}

In the above optimization problem on the unitary group, we considered the conjugation action of $\U(n)$ on $\mathcal{H}(n)$, the space of Hermitian matrices equipped with a particular inner product. 
To generalize this to other Lie groups beyond $\U(n)$, an easy observation is that $\mathcal{H}(n) = i \cdot \mathfrak{u}_n$, and thus up to this scaling by $i$ we can view $\mathcal{H}(n)$ as essentially equal to the Lie algebra $\mathfrak{u}_n$ of $\U(n)$.

The next question is: What then is the Lie-theoretic equivalent of the Frobenius inner product? 
To answer this, let us recall two properties of the Frobenius inner product which we discussed above. Specifically, for any $X,Y \in \mathfrak{u}_n$ and any $U \in \U(n)$, we have
\[
    \langle \Ad_U(X), \Ad_U(Y) \rangle_F = \Tr(UXU^{-1}UYU^{-1}) = \Tr(UXYU^{-1}) = \Tr(XY) = \langle X, Y \rangle_F
\]
and
\[
\begin{split}
    \langle \ad_X(Y), Z \rangle_F &= \Tr((XY-YX)Z) = \Tr(XYZ-YXZ) = \Tr(YXZ-XYZ) \\
        &= -\Tr(YXZ-YZX) = -\Tr(Y(XZ-ZX)) = -\langle Y, \ad_X(Z) \rangle_F.
\end{split}
\]
That is, the Frobenius inner product is invariant under the group and algebra adjoint actions of $\U(n)$ and $\mathfrak{u}_n$ on $\mathfrak{u}_n$. 
Note that the Frobenius ``inner product'' is actually \emph{negative} definite when applied to $\mathfrak{u}_n$ since $\mathfrak{u}_n$ consists of \emph{skew}-Hermitian matrices.

How can we generalize this to other Lie groups and algebras? 
In the general case, a Lie algebra $\mathfrak{g}$ is not necessarily a vector space of matrices, and so we cannot apply the trace. However, we know that $\ad_X$ is a linear operator on $\mathfrak{g}$, and thus can be considered as a matrix according to some chosen basis. With this, we can define a Frobenius-like symmetric bilinear form $B$ on $\mathfrak{g}$ in general, via:
\[
    B(X, Y) := \Tr(\ad_X \circ \ad_Y).
\]
We now attempt to prove some invariance properties similar to that of the Frobenius inner product. 
First, let us determine how $\Ad_g$ and $\ad_Z$ act on the linear map $\ad_X$. First we have
\[
\begin{split}
    \ad_{\Ad_g(X)}(Y) &= [gXg^{-1}, Y] = gXg^{-1}Y - YgXg^{-1} \\
        &= g(Xg^{-1}Yg - g^{-1}YgX)g^{-1} = [\Ad_g \circ \ad_X \circ \Ad_g^{-1}](Y),
\end{split}
\]
which says that $\ad_{\Ad_g(X)} = \Ad_g \circ \ad_X \circ \Ad_g^{-1}$. And second we have
\[
    \ad_{\ad_Z(X)}(Y) = [[Z,X],Y] = [Z,[X,Y]] - [X,[Z,Y]] = [\ad_Z \circ \ad_X - \ad_X \circ \ad_Z](Y),
\]
which says that $\ad_{\ad_Z(X)} = \ad_Z \circ \ad_X - \ad_X \circ \ad_Z$, the matrix commutator of $\ad_Z$ and $\ad_X$. 
This says that $\Ad_g$ and $\ad_Z$ act on $\ad_X$ via the adjoint actions of $\GL_n$ and $\mathfrak{gl}_n$ on $\ad_X \in \mathfrak{gl}_n$. 
The same arguments as in the case of $\U(n)$ and $\mathfrak{u}_n$ then imply the desired invariance properties. 
Specifically, we have
\[
    B(\Ad_g(X), \Ad_g(Y)) = B(X,Y) \qquad \text{and} \qquad B(\ad_Z(X), Y) = -B(X, \ad_Z(Y)).
\]
This symmetric bilinear form is called the \textbf{Killing form}. 
For semisimple matrix groups, this Killing form can be explicitly expressed with reference to the Frobenius inner product, via
\[
    B(X, Y) = C_{\mathfrak{g}} \cdot \Tr(X, Y),
\]
where $C_{\mathfrak{g}}$ is some positive constant which depends on the Lie algebra $\mathfrak{g}$ being considered. 

\begin{remark}
    The expression $B(X, Y) = C_{\mathfrak{g}} \cdot \Tr(X, Y)$ holds for the semisimple Lie algebras $\mathfrak{su}_n$, $\mathfrak{sl}_n$, $\mathfrak{so}_n$, and $\mathfrak{sp}_{2n}$, among others. These expressions do not quite hold for $\mathfrak{u}_n$ and $\mathfrak{gl}_n$ because they have nontrivial center. 
    In fact, the Killing form is identically zero on the center of $\mathfrak{g}$, but this point is not be very important to our discussion, see Remark \ref{rem:Killing_u_n}.
\end{remark}
As we saw with $\mathfrak{u}_n$, this bilinear form is not an inner product on $\mathfrak{g}$.
However, for a compact Lie algebra the Killing form is always negative semidefinite. 
And further, for a compact semisimple Lie algebra the Killing form is negative definite, and thus the negation of the Killing form gives rise to an $\Ad_g$-invariant and $\ad_Z$-invariant inner product on $\mathfrak{g}$ in this case.

\subsection{Maximal Tori and Cartan Subalgebras}

The Killing form defined above gives us a way to generalize the Frobenius inner product with its invariance properties to other Lie algebras. 
We next generalize the notion of diagonal matrix to a general Lie algebra.

Given a Lie algebra $\mathfrak{g}$, there exist certain subalgebras $\mathfrak{t} \subset \mathfrak{g}$ called Cartan subalgebras. 
While such subalgebras are typically defined for any Lie algebra, we only define them for compact Lie algebras here.

\begin{definition}[Cartan subalgebra]
    Given a compact Lie algebra $\mathfrak{g}$, a \textbf{Cartan subalgebra} or \textbf{maximal abelian subalgebra} $\mathfrak{t} \subset \mathfrak{g}$ is an abelian subalgebra of $\mathfrak{g}$ which is not properly contained in any other abelian subalgebra of $\mathfrak{g}$.
\end{definition}
In the literature, the term Cartan subalgebra is typically used in reference to complex Lie groups, whereas often the term \textbf{maximal abelian subalgebra} is used in reference to compact Lie groups. 
We also note here a standard result regarding Cartan subalgebras of compact Lie algebras.

\begin{proposition}[Prop. 4.30 of \cite{knapp2013}] \label{prop:Cartan_torus}
    If $G$ is a compact group with compact Lie algebra $\mathfrak{g}$, then $\mathfrak{t} \subset \mathfrak{g}$ is a Cartan subalgebra if and only if $\mathfrak{t}$ is the Lie subalgebra associated to a maximal torus $T \subset G$.
\end{proposition}
To get a more concrete feel for the Cartan subalgebra, let us consider $\mathfrak{t} \subset \mathfrak{u}_n$ to be the set of all diagonal matrices in $\mathfrak{u}_n$. 
For all $X,Y \in \mathfrak{t}$, we have that
\[
    [X,Y] = XY-YX = 0
\]
since $X$ and $Y$ are diagonal matrices. This implies $\mathfrak{t}$ is an abelian subalgebra of $\mathfrak{u}_n$. 
Further, it is straightforward to see that for all $Z \in \mathfrak{u}_n$ which is not diagonal, there is some diagonal $X \in \mathfrak{u}_n$ for which $[X,Z] \neq 0$. 
This implies $\mathfrak{t}$ is in fact a maximal abelian subalgebra of $\mathfrak{u}_n$. The Lie subgroup $T \subset \U(n)$ associated to $\mathfrak{t}$ as described by Proposition \ref{prop:Cartan_torus} above is then the compact torus of all diagonal matrices in $\U(n)$.

\begin{remark}
    The usual purpose of a Cartan subalgebra $\mathfrak{h}$ of a Lie algebra $\mathfrak{g}$ in the complex semisimple case is for the classification of semisimple Lie algebras, along with their representations and the representations of their associated Lie groups.
    The notion of a group representation was defined in Section \ref{sec:symmetries_revisited}, and a Lie algebra representation can be defined via the differential of a group representation.
    The main point is that the linear operators $\ad_H$ for $H \in \mathfrak{h}$ are all simultaneously diagonalizable. 
    From this, one can construct the \textbf{root system} of the Lie algebra $\mathfrak{g}$, which is a combinatorial construction which helps with Lie algebra classification and which we do not describe any further here. 
    When considering a representation $V$ of a given Lie algebra $\mathfrak{g}$, the action of the Cartan subalgebra $\mathfrak{h}$ via linear operators on $V$ is also simultaneously diagonalizable. 
    The associated eigenvalues can then be used to determine the decomposition of $V$ into irreducible subrepresentations via \textbf{highest weight theory}.
    We do not explore any of this further here, but the interested reader can find detailed explanations in any standard reference, e.g., \cite{baker2003matrix,knapp2013,hall2003lie}.
\end{remark}

\subsection{The Weyl Group}\label{sec:Weyl}

The Killing form defined above gives us a way to generalize the Frobenius inner product with its invariance properties to other Lie algebras, and the Cartan subalgebras of the previous section generalize the notion of diagonal matrices. 
The last thing we need to discuss and generalize is the symmetric group that appeared in our simple optimization problem above.

Associated to every semisimple or compact Lie group $G$ is a particular finite group $W$ called its Weyl group, and in the case of $\U(n)$, this is precisely the symmetric group $S_n$. 
For us, the key intuition behind this notion is that the Weyl group connects the continuous symmetries of a Lie group $G$ and its Lie algebra $\mathfrak{g}$ to underlying and hidden discrete finite symmetries of $G$ and $\mathfrak{g}$. 
This is  made very clear when we reinterpret our simple optimization problem in the language of Lie theory in the next section.

There are different ways to define the Weyl group of a given Lie group, but we describe just one here which specifically pertains to connected compact Lie groups.

\begin{definition}[Weyl group]
    Let $G$ be a connected compact Lie group, and let $T \subset G$ be a maximal torus in $G$. Let the normalizer of $T$ in $G$ be defined as usual, via $N(T) := \{g \in G : gtg^{-1} \in T, ~ \forall t \in T\}$. 
    The \textbf{Weyl group} of $G$ is defined to be the quotient group $W := N(T)/T$.
\end{definition}
Although it is not immediately clear from the definition, the definition of the Weyl group is independent of the choice of maximal torus $T \subset G$, and additionally the Weyl group is always a finite group. 
To give intuition for the fact that $W$ is finite, we explicitly construct it in the case of the unitary group as follows. 
First, we let $T \subset \U(n)$ be the subgroup of all diagonal unitary matrices; we know this is a maximal torus in $\U(n)$ for the discussion of the previous section. 
The normalizer $N(T)$ of $T$ is then equal to the set of matrices of the form $PD$, where $P$ is a permutation matrix and $D$ is a diagonal unitary matrix. 
From this description, it is straightforward to see that $N(T)/T$ is the set of permutation matrices, which is isomorphic to the symmetric group $S_n$.

\subsection{The Kostant Convexity Theorem}

We now have all the Lie-theoretically generalized pieces from our simple optimization problem, and we can now rewrite it in these generalized terms. 
Let $G$ be a connected compact Lie group with compact semisimple Lie algebra $\mathfrak{g}$, let $T \subset G$ be a maximal torus with associated Lie algebra $\mathfrak{t}$, and let $B$ be the Killing form of $\mathfrak{g}$. 
Our optimization problem can then be generalized as follows. Given $X,Z \in \mathfrak{t}$ with $\mathcal{O}_X$ being the adjoint orbit of $X$ in $\mathfrak{g}$, we want to optimize:
\[
    \min_{Y \in \mathcal{O}_X} -B(X, Z).
\]
Recall that the Killing form $B$ is negative definite since $G$ is compact semisimple.

But how can we solve this optimization problem? 
In the original example, the explicitness of the formulation gave us a few simple answers to this question, via the diagonal entries and the eigenvalues. 
In order to generalize these techniques to the above Lie-theoretic version of the problem, we need the following theorem which demonstrates the hidden convexity of the adjoint orbits of the Lie group $G$.

\begin{theorem}[Kostant convexity theorem, \cite{Kostant1973}; see also \cite{ziegler1992kostant}]\label{thm:Kostant}
    Let $G$ be a connected compact Lie group with compact semisimple Lie algebra $\mathfrak{g}$. 
    Let $T \subset G$ be a maximal torus with associated Lie algebra $\mathfrak{t}$. 
    And let $\pi: \mathfrak{g} \to \mathfrak{t}$ be the orthogonal projection induced by the Killing form $B$ (which is negative definite since $\mathfrak{g}$ is semisimple). 
    Then for every $X \in \mathfrak{t}$, the (group) adjoint orbit $\mathcal{O}_X \subset \mathfrak{g}$ of $G$ intersects $\mathfrak{t}$ in a Weyl group orbit $W \cdot X$, and $\pi(\mathcal{O}_X) = \hull(W \cdot X)$.
\end{theorem}

\noindent
We now prove one direction of the Kostant convexity theorem for the unitary group.
The other direction relies on the weight theory of Lie algebras which we have not discussed in detail here.
A short proof using weight theory can be found in \cite{ziegler1992kostant}.
Note that although technically $\mathfrak{u}_n$ is not semisimple, we can still prove the result by using the slightly tweaked Killing form given by $\Tr(XY)$. (See Remark \ref{rem:Killing_u_n} for more discussion on the Killing form of $\mathfrak{u}_n$.) 
Although this proof is specifically for the unitary group and its Lie algebra, it captures the essence of the general proof for compact Lie groups and their Lie algebras.

\begin{corollary}[Kostant convexity theorem for $\U(n)$] \label{cor:kostant_unitary}
    Let $\mathfrak{u}_n$ be the Lie algebra of the Lie group $\U(n)$, let $\mathfrak{t} \subset \mathfrak{u}_n$ and $T \subset \U(n)$ be the respective subsets of diagonal matrices, and let $\pi: \mathfrak{u}_n \to \mathfrak{t}$ be the orthogonal projection induced by the inner product $-\Tr(XY)$. Then for every $X \in \mathfrak{t}$, the (group) adjoint orbit $\mathcal{O}_X \subset \mathfrak{u}_n$ of $\U(n)$ intersects $\mathfrak{t}$ in a Weyl group orbit $S_n \cdot X$, and $\pi(\mathcal{O}_X) \subseteq \hull(S_n \cdot X)$.
\end{corollary}
\begin{proof}
    First we show that $\mathcal{O}_X \subset \mathfrak{u}_n$ of $\U(n)$ intersects $\mathfrak{t}$ in a Weyl group orbit $S_n \cdot X$. 
    This follows from the fact that the diagonal matrices in the orbit $\mathcal{O}_X$ are precisely those diagonal matrices with the eigenvalues of $X$ along the diagonal. 
    Since $X$ is an element $\mathfrak{t}$ and thus is itself diagonal, we have that $\mathcal{O}_X \cap \mathfrak{t} = S_n \cdot X$.
    
    Since $S_n \cdot X \subset \mathcal{O}_X$, we further have that the extreme points of $\hull(S_n \cdot X)$ are contained in $\mathcal{O}_X$ and also in $\pi(\mathcal{O}_X)$, since $\pi$ acts by sending all off-diagonal entries to 0. 
    The statement $\pi(\mathcal{O}_X) \subseteq \hull(W \cdot X)$ is thus equivalent to saying that every linear functional $\ell(Y)$ on $\pi(\mathcal{O}_X) \subset \mathfrak{t}$ is minimized at a point of $W \cdot X$. 
    Further, any linear functional on $\mathfrak{t}$ can be written as $\ell(Y) := \Tr(YL)$ for some $L \in \mathfrak{t}$, and for such $L$ we also have that $\Tr(YL) = \Tr(\pi(Y)L)$ for all $Y \in \mathcal{O}_X$. Thus it suffices to show that $\ell(Y) := \Tr(YL)$ is minimized over $\mathcal{O}_X$ at a point of $S_n \cdot X$ for all $L \in \mathfrak{t}$. 
    In fact, it suffices to show this for all $L \in \mathfrak{t}$ with distinct diagonal entries, by limiting.
    
    To show this, let us now consider $\mathcal{O}_X$ as a manifold acted upon by $\U(n)$ via the (group) adjoint action. 
    The tangent space of $\mathcal{O}_X$ at a point $Y$ can then be determined by considering the differential of the group adjoint action, which as we saw above gives the algebra adjoint action (see Section \ref{sec:algebra_adjoint_action}). 
    That is, for any $Y \in \mathcal{O}_X$ the tangent space of $\mathcal{O}_X$ at $Y$ is given by
    \[
        T_Y \mathcal{O}_X := \{\ad_Z(Y) = [Z,Y] : Z \in \mathfrak{u}_n\}.
    \]
    Further, $Y$ is a critical point of $\ell(Y) = \Tr(YL)$ in $\mathcal{O}_X$ if and only if for all paths $f: (-\epsilon,\epsilon) \to \mathcal{O}_X$ such that $f(0) = Y$ we have $\left.\frac{d}{dt}\right|_{t=0} \ell(f(t)) = 0$. Since $\ell$ is linear, we have that
    \[
        \left.\frac{d}{dt}\right|_{t=0} \ell(f(t)) = \ell\left(\left.\frac{d}{dt}\right|_{t=0} f(t)\right) = \ell(\ad_Z(Y))
    \]
    for some $\ad_Z(Y) \in T_Y \mathcal{O}_X$, and every element of $T_Y \mathcal{O}_X$ is the derivative of a path $f$ at 0 in this way.
    Thus, $Y$ is a critical point of $\ell(Y) = \Tr(YL)$ in $\mathcal{O}_X$ if and only if
    \[
        0 = \ell(\ad_Z(Y)) = \Tr(\ad_Z(Y)L) = -\Tr(\ad_Y(Z)L) = \Tr(Z\ad_Y(L))
    \]
    for all $Z \in \mathfrak{u}_n$. 
    Since $-\Tr(XY)$ is an inner product on $\mathfrak{u}_n$, this implies $\ad_Y(L) = YL-LY = 0$. 
    By assuming that $L \in \mathfrak{t}$ has distinct diagonal entries as discussed above, this further implies that $Y$ must be diagonal by a straightforward computation.
    Therefore, $Y \in \mathcal{O}_X \cap \mathfrak{t} = S_n \cdot X$, and this completes the proof.
\end{proof}

\noindent
The key takeaway of the Kostant convexity theorem is the fact that the image of a nonconvex adjoint orbit of $\mathfrak{g}$ under a linear projection map is a convex polytope.
And a key corollary to this is the fact that linear optimization over the orbitope generated by the adjoint orbit $\mathcal{O}_X$ is equivalent to linear optimization over an associated convex polytope. 
An \textbf{orbitope} is defined to be the convex hull of an orbit of some group action on a vector space, and these objects were extensively studied in \cite{sanyal2011}. 
It is also worth noting that these objects are sometimes spectrahedral, especially in the context of Lie theory, see \cite{Kobertphdthesis}.

We now state formally the key corollary regarding transferring linear optimization from the orbitope of an adjoint orbit $\mathcal{O}_X$ to the associated polytope guaranteed by the Kostnat convexity theorem.

\begin{corollary}\label{cor:linearopt}
    Let $G$, $T$, $\mathfrak{g}$, $\mathfrak{t}$ be as in the Kostant convexity
theorem with $\mathfrak{t} \subset \mathfrak{g}$, and suppose $\mathfrak{g}$ is semisimple with negative definite Killing form $B$. Then for any $X,Z \in \mathfrak{t}$, we have
\[
    \min_{Y \in \mathcal{O}_X} -B(Y, Z) = \min_{Y \in W \cdot X} -B(Y, Z),
\]
where $\mathcal{O}_X$ is the adjoint orbit of $X$ and $W \cdot X$ is the Weyl group orbit of $X$. 
That is, optimization of a linear functional can be restricted to $W \cdot X = \mathcal{O}_X \cap \mathfrak{t}$ in this case.
\end{corollary}
\begin{proof}
    Let $\pi: \mathfrak{g} \to \mathfrak{t}$ be the orthogonal projection induced by the Killing form $B$, as in the Kostant convexity theorem. Since $Z \in \mathfrak{t}$, we have
    \[
        \min_{Y \in \mathcal{O}_X} -B(Y, Z) = \min_{Y \in \mathcal{O}_X} -B(\pi(Y), Z) = \min_{Y \in \pi(\mathcal{O}_X)} -B(Y, Z).
    \]
    The Kostant convexity theorem then implies
    \[
        \min_{Y \in \pi(\mathcal{O}_X)} -B(Y, Z) = \min_{Y \in \hull(W \cdot X)} -B(Y, Z).
    \]
    Since a linear functional is always optimized over a polytope at one of its extreme points, we finally have
    \[
        \min_{Y \in \hull(W \cdot X)} -B(Y, Z) = \min_{Y \in W \cdot Z} -B(Y,Z).
    \]
    This completes the proof.
\end{proof}

\begin{remark} \label{rem:Killing_u_n}
    Note that although we restrict to compact semisimple Lie algebras in many the above results, the same results can also be stated for more general compact Lie algebras by adjusting the Killing form $B$ on $\mathfrak{g}$ so that it is strictly negative definite but still $\Ad$-invariant and $\ad$-invariant. 
    To do this, one simply needs to construct any negated inner product $B_0$ on the abelian part of $\mathfrak{g}$, and consider $B + B_0$. 
    The invariance properties are then trivially satisfied by $B_0$ and thus also by $B + B_0$. This is precisely what was done in Corollary \ref{cor:kostant_unitary} to obtain the negated inner product $\Tr(XY)$ on $\mathfrak{u}_n$, since $\mathfrak{u}_n$ is not semisimple but has one-dimensional center given by purely imaginary multiples of the identity matrix. See \cite[Prop.~4.24]{knapp2013} for more information.
\end{remark}

\begin{remark}
    Applying Corollary \ref{cor:linearopt} to the unitary group via Corollary \ref{cor:kostant_unitary} gives a result about optimizing $\Tr(YZ)$ over the adjoint orbits contained in $\mathfrak{u}_n$, the space of \emph{skew}-Hermitian matrices. However, the goal of the optimization problem of Section \ref{sec:eigenvalue} is to optimize over orbits of the unitary group within the space of (non-skew) Hermitian matrices. To see that these problems are equivalent, note that $\mathfrak{u}_n = i \cdot \mathcal{H}_{n \times n}$ where $\mathcal{H}_{n \times n}$ is the (real) vector space of Hermitian matrices. And further, the subalgebra $\mathfrak{t}$ of diagonal skew-Hermitian matrices is the set of diagonal matrices with purely imaginary diagonal entries. Therefore optimizing $\Tr(YZ)$ over $Y \in \mathcal{O}_X \subset \mathfrak{u}_n$ for $Z \in \mathfrak{t}$ is equivalent to optimizing $\Tr(iY \cdot iZ) = -\Tr(YZ)$ where $Y$ lies in some conjugation orbit of Hermitian matrices and $Z$ is a Hermitian diagonal matrix. Since minimizing $\Tr(YZ)$ is equivalent to maximizing $-\Tr(YZ)$ and vice versa, this shows that the two optimization problems are equivalent.
\end{remark}
Beyond the Kostant convexity theorem, there are actually other convexity theorems regarding various types of manifolds. For example, the \textbf{ Atiyah-Guillemin-Sternberg theorem} \cite{Atiyah1982,guillemin1982convexity} gives a similar result in the case of a symplectic manifold and its corresponding moment map. 
In the case of a symplectic manifold being acted upon by a compact Lie group, the \textbf{Kirwan convexity theorem} \cite{kirwan1984convexity} gives a result similar to that of the Kostant convexity theorem. In fact this generalizes the Kostant convexity theorem, since the KKS symplectic structure of Remark \ref{rem:adjoint_orbits_KKS} shows that adjoint orbits of a compact Lie group are symplectic manifolds.

\begin{remark} \label{rem:coadjoint}
    The Kostant convexity theorem is actually originally stated for \emph{co}adjoint orbits in the Lie algebra dual space $\mathfrak{g}^*$. In our case, the invariant Killing form allows us to transfer this statement to adjoint orbits in order to state the theorem as we have above.
\end{remark}

\section{Lie Theory for Sampling over Orbits}

We now explore how symmetries produced by Lie theory can help with developing sampling algorithms on manifolds.
A standard hurdle in the development of sampling algorithms on finite discrete domains is the computation of the partition function.
Given a weighting of the points in a discrete set, the partition function is the sum of all weights, giving the normalization constant with which one can turn the weighting into a probability distribution. 
This computation often bounds the computational complexity of a sampling problem on the domain.
Computation of the partition function can sometimes be easy; for example in the case of the uniform distribution on the set of spanning trees on a connected graph, computation of the partition function is given by a determinant of a matrix with rows on the order of the number of vertices, even when the total number of trees is exponentially large (e.g., see \cite{feder1992balanced}).
On the other hand, computing the partition function of the uniform distribution on the set of perfect matchings of a bipartite graph is known to be equivalent to computing the permanent of a 0-1 matrix, which is a \#P-hard problem (e.g., see \cite{jerrum1989approximating}).

What if the sampling domain is a continuous manifold? 
In this case, the partition function becomes an integral rather than a sum, and the efficient computability of such a partition function becomes more difficult to obtain or understand.
Thus, if we want to sample from continuous Lie-theoretic domains like adjoint orbits, we need to be able to compute integrals over such domains. 
We soon see how Lie theory and symmetry make this possible.

\subsection{The Haar Measure}

Often computational questions regarding manifolds (such as integration, expectation, and sampling) require an appropriate probability measure on the manifold. 
In the discrete case or on convex bodies, there is an implicit probability measure that comes from the fact that there is a well-defined uniform measure. 
For manifolds in general, there is no clear way to define such a measure. 
In the case of compact Lie groups however, there is a natural invariant measure which we  call the uniform measure, and we define it now.

\begin{theorem}
    If $G$ is a compact Lie group, then there exists a unique probability measure $\mu$ on $G$ called the \textbf{Haar measure} or \textbf{uniform measure} which has the following properties:
    \begin{enumerate}
        \item \textbf{Probability measure:} $\int_G 1 \, d\mu(g) = 1$.
        \item \textbf{$G$-invariance:} For all $h \in G$ and all $\mu$-measurable functions $f$, we have:
        \[
            \int_G f(hg) \, d\mu(g) = \int_G f(g) \, d\mu(g) = \int_G f(gh) \, d\mu(g).
        \]
    \end{enumerate}
\end{theorem}

\begin{remark}
    Haar measures can actually be defined for Lie groups in general, but the definition is more delicate. 
    In particular, there are potentially two different Haar measures: one corresponding to invariance respect to left multiplication and the other with respect to right multiplication. 
    For compact groups these two notions coincide, and so we  restrict to this case here throughout.
\end{remark}
The Haar measure on a compact Lie group $G$ also induces Haar measures on more interesting manifolds which arise via certain constructions on $G$. 
For example, if a connected compact Lie group $G$ acts smoothly on some real manifold $M$, then the Haar measure on $G$ induces a $G$-invariant measure on any given orbit of the action of $G$ on $M$. This is formalized in the following.

\begin{lemma}[Theorem 1.9 of \cite{helgason1984}] \label{lem:quotient_integral}
    Let $G$ be a compact Lie group with compact Lie subgroup $H$. There exists a $G$-invariant measure on the manifold $G/H$ such that
    \[
        \int_G f(g) d\mu(g) = \int_{G / H} \left(\int_H f(gh) d\nu(h)\right) d\mu_x(g)
    \]
    for all $\mu$-measurable functions $f$ on $G$.
\end{lemma}

\begin{corollary} \label{cor:orbit_Haar_measure}
    Let $G$ be a compact Lie group, and let $G$ act smoothly on a manifold $M$. 
    Given $x \in M$, let $\mathcal{O}_x \subset M$ denote the orbit of $x$ in $M$ under the action of $G$. 
    There exists a unique $G$-invariant probability measure on $\mathcal{O}_x$ called the \textbf{Haar measure} or \textbf{uniform measure}.
\end{corollary}
\begin{proof}
    Since $G$ acts smoothly on $M$, the stabilizer $\Stab_x$ of $x$ in $G$ is a closed subgroup of $G$, and thus a closed Lie subgroup by Cartan's closed-subgroup theorem. 
    Therefore $\Stab_x$ is in fact a compact Lie group and thus has a Haar measure $\nu$, and we also have that $\mathcal{O}_x \cong G / \Stab_x$. 
    The previous lemma then implies there exists a unique $G$-invariant measure $\mu_x$ on $\mathcal{O}_x \cong G / \Stab_x$ such that
    \[
        \int_G f(g) d\mu(g) = \int_{G / \Stab_x} \left(\int_{\Stab_x} f(gh) d\nu(h)\right) d\mu_x(g) = \int_{\mathcal{O}_x} \left(\int_{\Stab_x} f(gh) d\nu(h)\right) d\mu_x(g).
    \]
    This measure $\mu_x$ is then precisely the claimed Haar measure on $\mathcal{O}_x$.
\end{proof}
The above theorem not only guarantees Haar measures on more general manifolds beyond compact Lie groups, but the proof also tells us how to write certain types of integrals on those manifolds in terms of integrals on the associated compact Lie group. 
This essentially says that computing integrals over such manifolds (adjoint orbits for example) boils down to computing a similar integral over the compact Lie group acting on the manifold.

\subsection{The Harish-Chandra and Harish-Chandra--Itzykson--zuber Integral Formulas}

We now have a notion of uniform measure with respect to which we can write down integrals over Lie groups and adjoint orbits. 
But, this seems to get us no closer to actually computing such integrals. 
This is where the symmetries of the uniform measure and the Lie group itself come into play, yielding two of the most important results (from a computational perspective) in the theory of compact Lie groups: the \textbf{Harish-Chandra (HC) integral formula}, along with its corollary, the \textbf{Harish-Chandra--Itzykson--Zuber (HCIZ) integral formula}.

\begin{theorem}[HC integral formula, \cite{HC}]\label{thm:HC}
    Let $G$ be a compact connected semisimple Lie group, $\mathfrak{g}$ the associated Lie algebra, and $\mathfrak{t} \subset \mathfrak{g}$ a Cartan subalgebra. 
    For any $H_1,H_2 \in \mathfrak{t}$, we have
    \[
        \int_G e^{B(H_1, \Ad_g(H_2))} \, d\mu(g) = C_G \cdot \frac{\sum_{w \in W} \epsilon(w) e^{B(H_1, w \cdot H_2)}}{\Pi(H_1) \Pi(H_2)},
    \]
    where $B$ is the Killing form of $\mathfrak{g}$, $\mu$ is the Haar measure on $G$, $W$ is the Weyl group of $G$, $\Pi$ is an efficiently computable function on $\mathfrak{t}$, $C_G$ is an efficiently computable constant, and $\epsilon(w)$ is the sign of $w$. (We do not define the notion of sign here, but note that it is analogous to the notion of the sign of a permutation in $S_n$.)
\end{theorem}

\noindent
A priori, without the HC formula, it is  surprising that the integral on the left hand side can be written in terms of  determinants.
However, with the expression in the HC formula and the knowledge of the Weyl group, one might guess that the right hand side may be written in terms of  determinants.

We now state the HCIZ formula, which can be viewed as a special case of the HC formula, where the Lie group is chosen to be the unitary group $\U(n)$. 

\begin{theorem}[HCIZ integral formula, \cite{HC,IZ1980}]\label{thm:HCIZ}
    Given real diagonal $n \times n$ matrices $Y = \diag(y)$ and $\Lambda = \diag(\lambda)$ with diagonal entries $y_1 > y_2 > \cdots > y_n$ and $\lambda_1 > \lambda_2 > \cdots > \lambda_n$ respectively, we have
    \[
        \int_{\U(n)} e^{-\langle Y, U \Lambda U^* \rangle_F} \, d\mu(U) = \left(\prod_{p=1}^{n-1} p!\right) \frac{\det([e^{-y_i\lambda_j}]_{1 \leq i,j \leq n})}{\prod_{i<j} -(y_i-y_j)(\lambda_i-\lambda_j)},
    \]
    where $\langle \cdot, \cdot \rangle_F$ is the Frobenius inner product and $\mu$ is the Haar measure on $\U(n)$.
\end{theorem}
\noindent
By expanding the determinant in the HCIZ formula, one obtains
\[
    \det([e^{-y_i\lambda_j}]_{1 \leq i,j \leq n}) = \sum_{\sigma \in S_n} \epsilon(\sigma) \prod_{i=1}^n e^{-y_i \lambda_{\sigma(i)}} = \sum_{\sigma \in S_n} \epsilon(\sigma) e^{-\langle y, \sigma \cdot \lambda \rangle} = \sum_{\sigma \in S_n} \epsilon(\sigma) e^{-\langle Y, \sigma \cdot \Lambda \rangle_F}
\]
since $Y$ and $\Lambda$ are diagonal, where $\epsilon(\sigma)$ denotes the sign of the permutation $\sigma$. This then implies
\[
    \int_{\U(n)} e^{-\langle Y, U \Lambda U^* \rangle_F} \, d\mu(U) = \left(\prod_{p=1}^{n-1} p!\right) \frac{\sum_{\sigma \in S_n} \epsilon(\sigma) e^{-\langle Y, \sigma \cdot \Lambda \rangle_F}}{\prod_{i<j} -(y_i-y_j)(\lambda_i-\lambda_j)},
\]
which is essentially the HC formula. Note that formally, the HCIZ formula is not an immediate corollary of the HC formula because $\U(n)$ is not semisimple. 
However, this derivation can be made formal; see \cite{mcswiggen2019}.

\begin{remark}
    The difference in various signs between the HC and HCIZ formulas comes from the fact that in the HCIZ formula we deal with Hermitian matrices, while in the HC formula for $G = \U(n)$ we deal with the Lie algebra $\mathfrak{u}_n$ consisting of \emph{skew}-Hermitian matrices.
\end{remark}
Harish-Chandra originally studied such integrals in the `50s in order to develop a theory of Fourier analysis on semisimple Lie algebras.
Since then, these formulas have acquired far-reaching applications in other areas. 
The integral expression given above in the HCIZ formula is a partition function for a distribution on the set of unitary matrices $\U(n)$, which we call \textbf{HCIZ densities}.
Such unitary matrices are distributed according to the density function $e^{-\langle Y, U \Lambda U^* \rangle_F}$, and they are important in various settings in physics and random matrix theory. 
For instance, they appear in multi-matrix models in quantum field theory and string theory \cite{IZ,DGZ}, and they are also related to models of coupled Gaussian matrices \cite{IZ} that have been used to solve the Ising model on a planar random lattice \cite{Kaz, BouKaz}. 
Further, the moments of such HCIZ-distributed matrices are useful for computing correlation functions for matrix models of gauge theories which have been studied for over 20 years \cite{Morozov, Sha, Eynard-note, E-PF, PEDZ}.

Beyond HCIZ-distributed random matrices, the integral itself occurs in expressions for joint spectral densities of certain matrix ensembles such as Wishart matrices and off-center Wigner matrices \cite{AGsurvey}. 
The integral is also important in representation theory, where the corresponding finite sum formula is closely related to the Weyl character formula which describes the characters of irreducible representations of compact Lie groups (see \cite{knapp2013}).

The algorithmic significance of these results is that they give discrete sum formulas for continuous partition functions on compact Lie groups. 
Conceptually, the Weyl group of a Lie group $G$ acts as a ``quadrature group'' for $G$; that is, integrals of certain exponential functions over the group $G$ can be computed by instead summing those functions over the finite Weyl group $W$. 
Unfortunately the Weyl group $W$ is often related to the symmetric group, and so it is a priori too large for the HC formula to be efficiently computable. 
For many matrix groups however, the formula can be rewritten in terms of a small sum of determinants, as in the HCIZ formula, and this leads to something which can be efficiently computed, see \cite{LeakeV20b}. 
See \cite{mcswiggen2018harish} for more details on the formula itself, as well as its history.

We now have formulas for integrals over a compact Lie group, but how do these HCIZ and HC densities and formulas relate to integral formulas on the adjoint orbits of a Lie group?
We now answer this explicitly in the case of unitary orbits. 
Once we have the HCIZ formula, we can apply Lemma \ref{lem:quotient_integral} and Corollary \ref{cor:orbit_Haar_measure} to immediately obtain a similar integral formula on the adjoint orbits of $\U(n)$, given as follows.
\[
    \int_{\mathcal{O}_\Lambda} e^{-\langle Y, X \rangle_F} \, d\mu_\Lambda(X) = \int_{\U(n)} e^{-\langle Y, U \Lambda U^* \rangle_F} \, d\mu(U) = \left(\prod_{p=1}^{n-1} p!\right) \frac{\det([e^{-y_i\lambda_j}]_{1 \leq i,j \leq n})}{\prod_{i<j} -(y_i-y_j)(\lambda_i-\lambda_j)}.
\]
These log-linear densities on the adjoint obrits of $\U(n)$ are also well-studied. 
They arise for example as solutions to certain maximum entropy problems over the adjoint orbits, with one specific application being the computation of the entropy-maximizing representation of a quantum density matrix as an ensemble of pure states \cite{Band1976,Park1977,SlaterEntropy1}; see also \cite{LeakeV20}.
When the adjoint orbit is given by rank-$k$ PSD projections, such distributions give rise to the exponential mechanism for differentially private low-rank approximation \cite{MTalwar,Kamalika,KTalwar}. 
In the context of statistics, these distributions go by the name \textbf{matrix Langevin} and \textbf{matrix Bingham} distributions, except that typically one considers adjoint orbits of the orthogonal group rather than the unitary group \cite{ChikusePaper,ChikuseBook}.
And finally, when considering a convex body instead of an adjoint orbit, such entropy-maximizing distributions arose in the work of Klartag (inspired by a work of Gromov) on the isotropic constant \cite{Klartag2006,Gromov1990}, and these distributions also appear connection to the works of G\"uler, Bubeck and Eldan on barrier functions for interior point methods \cite{Guler1997,Guler1998,BubeckEldan}.

These HC, HCIZ, and orbit formulas immediately prompt a sampling question. 
That is, we have explicit expressions for the partition functions of these densities on a compact Lie group; can we efficiently sample from such distributions?

In the case of the unitary group, the problem of sampling according to a given HCIZ density and the problem of sampling from to a specific adjoint orbit are equivalent. 
The sampling question in this case has been answered affirmatively, and an efficient algorithm for sampling has been recently given by \cite{LMV2021}. 
The algorithm makes crucial use of Lie theory: the idea of obtaining a convex polytope from projecting a given adjoint orbit, as in the Kostant convexity theorem, is an essential part of the algorithm. 
Generally speaking, one first samples from the polytope, and then samples from the adjoint orbit by sampling from the corresponding fiber of the projection map. (Though it should be mentioned that the polytope and projection are slightly different; they are the \textbf{Gelfand-Tsetlin polytope} and the \textbf{Rayleigh map}. See \cite{LMV2021} for more discussion.)

Beyond the unitary group, the question of efficient sampling remains open, even in the case of adjoint orbits of the real orthogonal group. That is, efficiently sampling from matrix Langevin and matrix Bingham distributions as they are traditionally defined is an open problem.

\begin{remark}
    The expression above for the integral of a log-linear function over $\mathcal{O}_\Lambda$ looks different than the expression of Theorem \ref{thm:partition}. This is because here we have assumed that the eigenvalues of $\Lambda$ are distinct, while the eigenvalues of $X = e_1e_1^\top$ are not distinct. 
    One way to make the connection between these two cases is via a limiting argument and applying L'H\^opital's rule. 
    The details can be found in \cite{LeakeV20}.
\end{remark}

\subsubsection{Proof (Sketch) of the HCIZ Formula}

In this section, we sketch a proof of the HCIZ formula (Theorem \ref{thm:HCIZ}) due to \cite{faraut2015rayleigh}. Recall that we want to prove
\[
    \int_{\U(n)} e^{-\langle Y, U \Lambda U^* \rangle_F} \, d\mu(U) = \left(\prod_{p=1}^{n-1} p!\right) \frac{\det([e^{-y_i\lambda_j}]_{1 \leq i,j \leq n})}{\prod_{i<j} -(y_i-y_j)(\lambda_i-\lambda_j)}.
\]
First recall that by Lemma \ref{lem:quotient_integral} and Corollary \ref{cor:orbit_Haar_measure}, we have that
\[
    \int_{\mathcal{O}_\Lambda} e^{-\langle Y, X \rangle_F} \, d\mu_\Lambda(X) = \int_{\U(n)} e^{-\langle Y, U \Lambda U^* \rangle_F} \, d\mu(U).
\]
We now define the measure $\mu_\Lambda^{(k)}$ to be the measure obtained by projecting $\mu_\Lambda$ through the map $\pi_k$ which maps a matrix to its leading $k \times k$ principal submatrix. Letting $\mathcal{O}_\Lambda^{(k)} := \pi_k(\mathcal{O}_\Lambda)$, this definition implies the formula
\[
    \int_{\mathcal{O}_\Lambda^{(k)}} f(Y) \, d\mu_\Lambda^{(k)}(Y) = \int_{\mathcal{O}_\Lambda} f(\pi_k(X)) \, d\mu_\Lambda(X)
\]
for any function $f$ on $\mathcal{O}_\Lambda^{(k)}$. 
The support $\mathcal{O}_\Lambda^{(k)}$ of the measure $\mu_\Lambda^{(k)}$ can be described explicitly, using the following result often attributed to Cauchy or Rayleigh.

\begin{proposition}[Rayleigh-Cauchy]
    Let $\Lambda$ be a Hermitian matrix with eigenvalues $\lambda_1 \geq \lambda_2 \geq \cdots \geq \lambda_n$, and let $\mu_1 \geq \cdots \geq \mu_{n-1}$ be the eigenvalues of the leading $(n-1) \times (n-1)$ principal submatrix of $\Lambda$. 
    Then the following interlacing relations hold:
    \[
        \lambda_1 \geq \mu_1 \geq \lambda_2 \geq \mu_2 \geq \cdots \geq \lambda_{n-1} \geq \mu_{n-1} \geq \lambda_n.
    \]
\end{proposition}

\noindent
This fact can actually be strengthened: in fact, these interlacing relations completely determine the support of the measure $\mu_\Lambda^{(n-1)}$:
\[
    \mathcal{O}_\Lambda^{(n-1)} = \{M \in \C^{(n-1) \times (n-1)} : \eig(M) \text{ interlace } \eig(\Lambda)\}.
\]
That is, every $(n-1) \times (n-1)$ Hermitian matrix with eigenvalues interlacing those of $\Lambda$ can appear as the principal submatrix of a unitary conjugate of $\Lambda$. 
Further, the set of possible eigenvalues of elements of $\mathcal{O}_\Lambda^{(n-1)}$ form a convex scaled hypercube given by the above inequalities, which we  denote by $P_\Lambda$. 
By induction, we obtain similar results for $\mathcal{O}_\Lambda^{(k)}$ for all $k$.

Beyond this expression for $\mathcal{O}_\Lambda^{(n-1)}$, Barychnikov has given a formula for the probability distribution on the eigenvalues of matrices which can appear in $\mathcal{O}_\Lambda^{(n-1)}$. 
This is given explicitly as follows.

\begin{theorem}[Barychnikov, see Theorem 1.3 of \cite{faraut2015rayleigh}]
    Fix a Hermitian matrix $\Lambda$, and let $X$ be sampled from $\mathcal{O}_\Lambda$ according to the uniform probability measure $\mu_\Lambda$. 
    Then $\eig(\pi_{n-1}(X))$ is a vector sampled from the convex scaled hypercube $P_\Lambda$ of vectors which interlace $\eig(\Lambda)$, according to a probability measure which is a density function $g_\Lambda$ times Lebesgue measure restricted to $P_\Lambda$, and the $g_\Lambda$ is explicitly given by
    \[
        g_\Lambda(x) = (n-1)! \frac{V_{n-1}(x)}{V_n(\eig(\Lambda))},
    \]
    where $V_n(x)$ is the Vandermonde determinant $V_n(x) = \prod_{i<j} (x_i-x_j)$.
\end{theorem}

\noindent
This then further implies another equivalent expression for orbital integrals, given by
\[
    \int_{\mathcal{O}_\Lambda^{(n-1)}} f(Y) \, d\mu_\Lambda^{(n-1)}(Y) = \int_{a \in P_\Lambda} \left[\int_{\mathcal{O}_{\diag(a)}} f(X) \, d\mu_{\diag(a)}(X)\right] g_\Lambda(a) \, da.
\]
That is, we have split the integral by first selecting a collection of eigenvalues $a$ according to $f_\Lambda(a) da$ on $P_\Lambda$, where $f_\Lambda$ is as defined in the previous theorem, and then integrating over the unitary orbit of $\diag(a)$ which is contained within $\mathcal{O}_\Lambda^{(n-1)}$.

We now have the necessary tools to prove an inductive formula for the desired orbital intergal over $e^{-\langle Y, X \rangle_F}$. We first give an alternative expression for the exponent $-\langle X, Y \rangle_F$ in the orbital integral. For this, we now assume that $Y = \diag(y)$ is diagonal, which is without loss of generality by unitary conjugation. Letting $X'$ and $Y'$ denote the leading principal submatrices of $X$ and $Y$ respectively, we then have
\[
    -\langle Y, X \rangle_F = -\Tr(YX) = -\Tr(Y'X') - y_n (\Tr(X)-\Tr(X')) = -\langle Y', X' \rangle_F - y_n (\Tr(\Lambda)-\Tr(X')).
\]
We now plug this expression into the orbital integral to obtain
\[
    \int_{\mathcal{O}_\Lambda} e^{-\langle Y, X \rangle_F} \, d\mu_\Lambda(X) = e^{-y_n \Tr(\Lambda)} \int_{\mathcal{O}_\Lambda} e^{y_n \Tr(X')} e^{-\langle Y', X' \rangle_F} \, d\mu_\Lambda(X).
\]
The integrand is then a function only of $X' = \pi_{n-1}(X)$ and not of $X$ itself. 
Therefore we can apply the above formulas for integration over $\mathcal{O}_\Lambda^{(n-1)}$ to obtain
\[
\begin{split}
    \int_{\mathcal{O}_\Lambda} e^{-\langle Y, X \rangle_F} \, d\mu_\Lambda(X) &= e^{-y_n \Tr(\Lambda)} \int_{\mathcal{O}^{(n-1)}_\Lambda} e^{y_n \Tr(Z)} e^{-\langle \pi_{n-1}(Y), Z \rangle_F} \, d\mu_\Lambda^{(n-1)}(Z) \\
        &= e^{-y_n \Tr(\Lambda)} \int_{\mathcal{O}^{(n-1)}_\Lambda} e^{y_n \Tr(Z)} e^{-\langle \pi_{n-1}(Y), Z \rangle_F} \, d\mu_\Lambda^{(n-1)}(Z) \\
        &= \int_{P_\Lambda} e^{y_n(a_1+\cdots + a_{n-1} - \Tr(\Lambda))} \left[\int_{\mathcal{O}_{\diag(a)}} e^{-\langle \pi_{n-1}(Y), Z \rangle_F} \, d\mu_{\diag(a)}(Z)\right] g_\Lambda(a) \, da.
\end{split}
\]
That is, if $\mathcal{I}_\Lambda(Y) := \int_{\mathcal{O}_\Lambda} e^{-\langle Y, X \rangle_F} \, d\mu_\Lambda(X)$ is the integral we want to compute, we have the inductive formula
\[
    \mathcal{I}_\Lambda(Y) = (n-1)!\frac{ e^{-y_n \Tr(\Lambda)}}{V_n(\eig(\Lambda))} \int_{P_\Lambda} \mathcal{I}_{\diag(a)}(\pi_{n-1}(Y)) \cdot e^{y_n(a_1+\cdots + a_{n-1})} V_{n-1}(a) \, da.
\]
Our candidate formula for $\mathcal{I}_\Lambda(Y)$ is then of course given by the desired HCIZ formula, where $\lambda$ is the vector of eigenvalues of $\Lambda$:
\[
    \mathcal{I}_\Lambda(Y) = \left(\prod_{p=1}^{n-1} p!\right) \frac{\det([e^{-y_i\lambda_j}]_{1 \leq i,j \leq n})}{\prod_{i<j} -(y_i-y_j)(\lambda_i-\lambda_j)} = \left(\prod_{p=1}^{n-1} p!\right) \frac{\det([e^{-y_i\lambda_j}]_{1 \leq i,j \leq n})}{V_n(-y) V_n(\eig(\Lambda))}.
\]
With this, one then can prove relatively straightforwardly that this expression holds inductively. 
The key idea towards this is to note that since $P_\Lambda$ is a convex scaled hypercube, the integral $\int_{P_\Lambda} f(a) \, da$ can be written as a product of integrals over each of the $n-1$ coordinates of $a \in P_\Lambda \subset \R^{n-1}$. We do this as follows, where $\lambda$ is the vector of eigenvalues of $\Lambda$ and $y' := (y_1,\ldots,y_{n-1})$ is the vector of eigenvalues of $\pi_{n-1}(Y)$ since $Y$ is diagonal:
\[
\begin{split}
    \int_{P_\Lambda} &\mathcal{I}_{\diag(a)}(\pi_{n-1}(Y)) \cdot e^{y_n(a_1+\cdots + a_{n-1})} V_{n-1}(a) \, da \\
        &= \int_{\lambda_2}^{\lambda_1} \int_{\lambda_3}^{\lambda_2} \cdots \int_{\lambda_n}^{\lambda_{n-1}} \mathcal{I}_{\diag(a)}(\pi_{n-1}(Y)) \cdot e^{y_n(a_1+\cdots + a_{n-1})} V_{n-1}(a) \, d_{a_{n-1}} \cdots da_2 \, da_1 \\
        &= \int_{\lambda_2}^{\lambda_1} \int_{\lambda_3}^{\lambda_2} \cdots \int_{\lambda_n}^{\lambda_{n-1}} \left(\prod_{p=1}^{n-2} p!\right) \frac{\det([e^{-y_i a_j}]_{1 \leq i,j \leq n-1})}{V_{n-1}(-y') V_{n-1}(a)} \cdot e^{y_n(a_1+\cdots + a_{n-1})} V_{n-1}(a) \, d_{a_{n-1}} \cdots da_2 \, da_1 \\
        &= \frac{\prod_{p=1}^{n-2} p!}{V_{n-1}(-y')} \int_{\lambda_2}^{\lambda_1} \int_{\lambda_3}^{\lambda_2} \cdots \int_{\lambda_n}^{\lambda_{n-1}} \det([e^{(y_n-y_i) a_j}]_{1 \leq i,j \leq n-1})\, d_{a_{n-1}} \cdots da_2 \, da_1 \\
\end{split}
\]
Note now that each integration is a linear operation, and that the determinant in the proposed expression for $\mathcal{I}_{\diag(a)}(\pi_{n-1}(Y))$ is multilinear in the coordinates of $a$. 
These integrals can be passed into the determinant to be applied separately to each of the rows of the input matrix. 
Combining this with the above inductive expression for $\mathcal{I}_\Lambda(Y)$ then yields the desired HCIZ formula.

\section{Classification Theorems for Lie Algebras} \label{sec:classification}

In Section \ref{sec:Lie}, we showed how a Lie algebra could be defined as the tangent space of a Lie group $G$ as the identity element $e$, equipped with a certain product-like operation called the Lie bracket. We now define the notion of a Lie algebra more abstractly, and show how this notion coincides with that of Section \ref{sec:lie_alg_via_group}. 
The power of this abstract definition is that it makes no reference to Lie groups. 
Powerful classification theorems of Lie algebras, a priori having nothing to do with groups, then extend immediately to certain important classes of Lie groups.

Before moving on, we make one caveat. 
An abstract Lie algebra $\mathfrak{g}$ is a vector space equipped with a Lie bracket which we discuss below. 
The vector space structure of $\mathfrak{g}$ requires a choice of field, and for us this field is  either $\C$ or $\R$. 
Many of the nice classification results for Lie algebras use $\C$ because they rely on the fact that $\C$ is algebraically closed. 
However, compact Lie groups, like the unitary group, have associated Lie algebra which is necessarily real. 
(Even though $\mathfrak{u}_n$ consists of skew-Hermitian matrices with complex entries, there is no way to consider of the space of skew-Hermitian matrices as a complex vector space. For example, multiplication by $i$ does not preserve $\mathfrak{u}_n$.) 
This point is not very important to our exposition, but it is still worth making to avoid certain confusions.

\subsection{Abstract Lie Algebras}

An abstract Lie algebra $\mathfrak{g}$ is a vector space equipped with a product-like operation $[\cdot, \cdot]: \mathfrak{g} \times \mathfrak{g} \to \mathfrak{g}$ called the \textbf{Lie bracket}. 
We also  often write $\ad_X(Y) = [X,Y]$, where $\ad_X$ is called the (algebra) adjoint action of the Lie algebra on itself, as discussed above. 
The Lie bracket has three defining properties, given as follows:
\begin{enumerate}
    \item \textbf{Bilinearity:} $[aX+bY,Z] = a[X,Z] + b[Y,Z]$ and $[X,bY+cZ] = b[X,Y] = c[X,Z]$.
    \item \textbf{Anti-symmetry:} $[X,Y] = -[Y,X]$.
    \item \textbf{Jacobi identity:} $[X,[Y,Z]] + [Y,[Z,X]] + [Z,[X,Y]] = 0$.
\end{enumerate}
In the case of Lie algebras associated to matrix groups, the Lie bracket is given by the commutator $[X,Y] = XY-YX$. 
It is straightforward to see that this satisfies the defining properties listed above.

\subsection{Simple, Semisimple, and Reductive Lie Algebras}

We now discuss a number of important notions regarding a Lie algebra $\mathfrak{g}$, some of which we have already seen in the above discussion. 
This  helps us in our attempt to classify the Lie algebras we hope to consider. 
Here are some of them:
\begin{itemize}
    \item \textbf{Abelian:} A Lie algebra $\mathfrak{g}$ for which $[X,Y] = 0$ for all $X,Y \in \mathfrak{g}$.
    \item \textbf{Subalgebra:} A subspace $\mathfrak{h} \subset \mathfrak{g}$ for which $[X,Y] \in \mathfrak{h}$ for all $X,Y \in \mathfrak{h}$.
    \item \textbf{Ideal:} A subalgebra $\mathfrak{h} \subset \mathfrak{g}$ for which $[X,Y] \in \mathfrak{h}$ for all $X \in \mathfrak{g}$ and $Y \in \mathfrak{h}$.
    \item \textbf{Direct sum:} A direct sum of vector spaces $\mathfrak{g} = \bigoplus_{i=1}^m \mathfrak{g}_i$ with Lie bracket given by
    \[
        [(X_1,X_2,\ldots,X_m),(Y_1,Y_2,\ldots,Y_m)] := ([X_1,Y_1],[X_2,Y_2],\ldots,[X_m,Y_m]).
    \]
    \item \textbf{Simple:} A non-abelian Lie algebra which has no nontrivial ideals.
    \item \textbf{Semisimple:} A direct sum of simple Lie algebras.
    \item \textbf{Reductive:} A direct sum of a semisimple Lie algebra and an abelian Lie algebra.
\end{itemize}
As a note, abelian Lie algebras are immediately classified by the dimension of the underlying vector space. 
The classification of semisimple and reductive Lie algebras is more complicated, and we describe it in the next section.

\subsection{Classification of Complex Reductive Lie Algebras}

One of the most important features of Lie theory is that there is a bijection between a certain class of nice Lie groups and the important class of semisimple Lie algebras. 
That is, not only can one construct a Lie algebra from such a Lie group, but there is also an inverse map from abstractly-defined Lie algebras to nice Lie groups. 
Therefore such Lie groups can be classified by classifying Lie algebras, without any reference to the fact that these algebras come from groups. 
This is a powerful fact since Lie algebras are linear-algebraic objects, and linear algebra is typically much easier than the geometry that would come with studying the groups directly.

Broadly speaking, the classification relies upon the fact that a given Lie algebra can be associated to certain \emph{discrete} combinatorial objects called Dynkin diagrams. 
By classifying the valid combinatorial objects, one then obtains a classification of the associated Lie algebras. 
We do not go through such a proof of the classification, and the interested reader can find various proofs in standard references. Instead, we simply state the classification theorem.

\begin{theorem}[Classification of complex simple Lie algebras, see \S VI.10 of \cite{knapp2013}]
    Every complex simple Lie algebra fits into one of the following categories:
    \begin{itemize}
        \item $A_n \colon \mathfrak{sl}_{n+1}$, the Lie algebra associated to $\SL_n(\C)$.
        \item $B_n \colon \mathfrak{so}_{2n+1}$, the odd-dimensional Lie algebra associated to $\SO_{2n+1}(\C)$.
        \item $C_n \colon \mathfrak{sp}_{2n}$, the Lie algebra associated to $\Sp_{2n}(\C)$.
        \item $D_n \colon \mathfrak{so}_{2n}$, the even-dimensional Lie algebra associated to $\SO_{2n}(\C)$.
        \item $E_6, E_7, E_8, F_4, G_2 :$ the so-called \textbf{exceptional Lie algebras}.
    \end{itemize}
    Note that this also classifies complex semisimple and reductive Lie algebras by taking direct sums.
\end{theorem}

\noindent
We do not go into any more detail on the exceptional Lie algebras mentioned in the classification theorem. 
And as a final note, the notation used for the here ($A_n$, $B_n$, etc.) comes from the notation used to represent the \textbf{Dynkin diagrams} associated to the respective Lie algebras. 
A Dynkin diagram gives a way to encode the information of a Lie algebra in a purely combinatorial way, and classifying Dynkin diagrams is at the heart of most proofs of the above classification. 
We do not discuss Dynkin diagrams any further here.

\subsection{Classification of Compact Lie Algebras}

Although we have classified the complex semisimple Lie algebras above, this is not quite enough for our purposes. 
We actually want to classify Lie algebras associated to compact Lie groups. 
And as mentioned above, compact Lie groups must be considered as real manifolds; they cannot in general be given the structure of a complex manifold. 
Because of this, their associated Lie algebras are real instead of complex, and so they do not fit into the above classification.

That said, recall that we defined a compact Lie algebra to be a Lie algebra which is derived from a compact Lie group. 
Although this is an extrinsic definition, compactness of a Lie algebra $\mathfrak{g}$ can be defined intrinsically via properties of the Killing form of $\mathfrak{g}$, without any reference to an associated Lie group. 
That is, a compact Lie algebra is a Lie algebra from which the Killing form is negative semidefinite.

We now extend the above classification to cover compact Lie algebras. 
This classification requires the notion of a \textbf{real form} of a complex Lie algebra, which is a real Lie algebra whose complexification is the given complex Lie algebra. 
In fact, every semisimple complex Lie algebra has a unique compact real form, and the converse is also true, as stated in the following result.

\begin{theorem}[Classification of compact simple Lie algebras, see \S VI.10 of \cite{knapp2013}]
    Every compact simple Lie algebra fits into one of the following categories:
    \begin{itemize}
        \item $A_n \colon \mathfrak{su}_{n+1}$, the Lie algebra associated to $\SU(n+1)$.
        \item $B_n \colon \mathfrak{so}_{2n+1}$, the odd-dimensional Lie algebra associated to $\SO_{2n+1}(\R)$.
        \item $C_n \colon \mathfrak{usp}_n$, the Lie algebra associated to $\USp(n)$.
        \item $D_n \colon \mathfrak{so}_{2n}$, the even-dimensional Lie algebra associated to $\SO_{2n}(\R)$.
        \item $E_6, E_7, E_8, F_4, G_2 :$ the compact forms of the exceptional Lie algebras.
    \end{itemize}
    Note that this also classifies compact semisimple and reductive Lie algebras by taking direct sums.
\end{theorem}

Note that some of the notation in the complex and compact classifications results is overloaded. 
These are distinguished by the fact that the underlying vector spaces are complex and real vector spaces respectively, and this distinction is often clear from context.

The final point to make to classify compact Lie algebras completely is to handle the case of compact Lie algebras which are not reductive.
While there are many complex Lie algebras which are not reductive, the opposite is true in the compact case.

\begin{proposition}[Cor. 4.25 of \cite{knapp2013}]
    Every compact Lie algebra is reductive.
\end{proposition}

\noindent
This last result then completes the classification of compact Lie algebras.

\section{Summary and Conclusion}
This article provides both the motivation and the basics to leverage symmetries in the design of efficient algorithms.
In the first part, we discussed two basic and important problems. 
The first was an optimization characterization of the minimum eigenvalue problem (Theorem \ref{thm:eigenvalue}). 
For an $n \times n$ Hermitian matrix, we showed that
\[
    \min_{X \in \mathcal{P}_1} \langle A, X \rangle_F = \min_{U \in \U(n)} \langle A, Ue_1e_1^\top U^* \rangle_F = \lambda_1.
\]
This is a standard result, but we demonstrated how it follows from the Schur-Horn theorem. 
The Schur-Horn theorem says that the projection onto the diagonal of all matrices in the conjugation orbit of a Hermitian matrix yields a convex polytope. 
In the case of the orbit of $e_1e_1^\top$ considered above, this polytope is precisely the standard simplex corresponding to the convex hull of $\{e_{\sigma(1)} : \sigma \in S_n\}$. 
We further saw that this polytope is more generally given by a similar expression; that is, given a Hermitian matrix $A = \diag(v)$, the polytope is given by $\{\sigma \cdot v : \sigma \in S_n\}$ where $\sigma \cdot v$ acts by permuting the entries of $v$.

The second problem we considered was that of sampling from a certain exponential density on the complex unit sphere. 
We showed that this problem is connected to the first one: computing integrals and sampling from the unit sphere is essentially equivalent to doing the same things on $\mathcal{P}_1$, the conjugation orbit of $e_1e_1^\top$ by the unitary group. 
We then discussed the partition function associated to such a density, which is given by the integral of an exponential function over the orbit $\mathcal{P}_1$. 
Computing such partitions functions is often the first step in constructing a sampling algorithm, and usually the complexity of this computation determines the complexity of the sampling algorithm.

In this case of the partition function being an integral over a nonconvex manifold, it is not immediately clear how to compute such a partition function, and even less so if one wants to do it efficiently. 
In spite of this, we were able to show that not only could this partition function be efficiently computed in this case, but it in fact has a closed-form formula (Theorem \ref{thm:partition}):
\[
    \frac{1}{(n-1)!} \int_{\mathcal{P}_1} e^{-\langle Y, X \rangle_F} d\mu_1(X) = \sum_{i=1}^n \frac{e^{-\lambda_i}}{\prod_{j \neq i} (\lambda_j-\lambda_i)},
\]
where $\lambda_i$ are the eigenvalues of $Y$. 
We gave a proof of this fact, which used the key fact that computing certain integrals over the nonconvex manifold $\mathcal{P}_1$ is equivalent to computing similar integrals over the standard simplex.

In the second part, we focused on understanding very important groups which arise when there are continuous symmetries. 
These important groups were the Lie groups, which are special manifolds that have a compatible group structure. 
A crucial fact regarding Lie groups is that the tangent space at the identity recovers a large amount of information about the Lie group. 
This tangent space is a vector space equipped with a product-like operation coming from the group product, and it is called the Lie algebra.
The Lie algebra of a Lie group is linear-algebraic in nature, while the Lie group is more geometric. 
Thus the Lie algebra gives access to more linear-algebraic methodologies and algorithms.

The symmetries of a Lie group extend to the Lie algebra, and one way this is formalized is in the notion of the adjoint action of a Lie group on its Lie algebra. 
In the case of matrix groups, this adjoint action is precisely given by the natural action of matrix conjugation. 
This action allows Lie algebras to be viewed as a representation of the corresponding Lie group - called the ``adjoint representation.'' 
While we did not discuss it here, one can study more general representations of the Lie group via connections to the Lie algebra. 

The orbits of this adjoint action on the Lie algebra, called the adjoint orbits, are important Lie theoretic objects that we then studied further. 
These orbits are nonconvex manifolds, but there are interesting convex bodies associated to these objects. 
Convex hulls of adjoint orbits are called orbitopes \cite{sanyal2011}, and these are sometimes spectrahedral \cite{Kobertphdthesis}. But more interestingly, we saw that the image of a given adjoint orbit under the projection map to a Cartan subalgebra is a convex polytope, by the Kostant convexity theorem. 
In the case of the unitary group $\U(n)$, this yields the Schur-Horn theorem, which says that the set of all diagonal vectors of Hermitian matrices with given eigenvalues form a convex polytope. 
We then explored how certain optimization problems on nonconvex adjoint orbits could be reduced to optimization problems on the corresponding convex polytope. 
And, we also briefly mentioned that moment maps arise from symplectic structure via the KKS form on the adjoint orbits, which allows for Hamiltonian dynamics on these orbits.

Finally, we discussed the HC and HCIZ formulas for certain exponential integrals over compact Lie groups. 
These important results give formulas for partition functions of certain measures over compact Lie groups and their adjoint orbits. 
The key surprising fact regarding these formulas is that they serve to reduce integrals over a group to a finite sum, or even a few determinants. 
That is, they reduce continuous symmetries to discrete symmetries.

To conclude,  computational problems that arise in nonconvex optimization and sampling can benefit when viewed through the lens of symmetries.
This point of view has seen a recent surge and found several  applications  in  areas such as  complexity theory, fast graph algorithms, statistics,  machine learning, quantum inference, and differential privacy (see, e.g., \cite{GargGOW16,BurgisserFGOWW19,LeakeV20,LMV2021} and the references therein).
However, the full potential of  symmetries in optimization and sampling remains to be explored.

\section*{Acknowledgments}
The authors would like to thanks Anay Mehrotra and Yikai Wu for useful comments. NV would like to acknowledge the support of  NSF CCF-1908347.

\bibliographystyle{alpha}
\bibliography{lie}

\end{document}